\documentclass[imslayout,noinfoline,a4paper,12pt,twoside]{imsart}
\usepackage{diagrams}
\RequirePackage[OT1]{fontenc} 
\RequirePackage{amsthm,amsmath,amssymb,natbib} 

\startlocaldefs
\theoremstyle{plain}
\newtheorem{thm}{Theorem}[section]

\newtheorem{lem}{Lemma}[section]

\newtheorem{defin}{Definition}[section]
\endlocaldefs

\renewcommand{\[}{\begin{eqnarray*}}
\renewcommand{\]}{\end{eqnarray*}}
\newcommand{\la}{\begin{eqnarray}}
\newcommand{\al}{\end{eqnarray}}

\renewcommand{\epsilon}{\varepsilon}
\renewcommand{\phi}{\varphi}

\newcommand{\N}{{\mathbb N}}
\newcommand{\Q}{{\mathbb Q}}\newcommand{\R}{{\mathbb R}}
 
\newcommand{\cM}{{\mathcal M}}
\newcommand{\cP}{{\mathcal P}}\newcommand{\cQ}{{\mathcal Q}}
\newcommand{\cX}{{\mathcal X}}\newcommand{\cY}{{\mathcal Y}}

\newcommand{\Prob}{\mathrm{Prob}}\newcommand{\prob}{\mathrm{prob}}
\newcommand{\markov}{\mathrm{mark}}

\newcommand{\Eta}{\mathrm{H}}
\renewcommand{\theta}{\vartheta}

\newcommand{\defn}[1]{{\em#1}}

\renewcommand{\Pr}{\mathrm{Pr}}%
\newcommand{\Se}{\mathrm{Se}}  %
\newcommand{\Sp}{\mathrm{Sp}}  %

\usepackage{Sweave}
\begin{document}

\begin{frontmatter}


\title{Confidence bounds for the sensitivity lack of a less specific
diagnostic test, without gold standard}


\runtitle{Confidence bounds  for a  sensitivity lack, \today}

\begin{aug}
\author{\fnms{Lutz} \snm{Mattner}
\ead[label=e1]{mattner@uni-trier.de}} 
 \and
 \author{\fnms{Frauke} \snm{Mattner}
 \ead[label=e2]{mattnerf@kliniken-koeln.de}}%

\runauthor{L.~Mattner and F.~Mattner, \today}

\affiliation{Universit\"at Trier and Kliniken der Stadt K\"oln} 
 {\rm \today}\\
 {\footnotesize\tt \jobname.tex}

\address{Universit\"at Trier\\
Fachbereich IV - Mathematik\\ 
54286 Trier \\
Germany\\
\printead{e1}}
\address{
Universit\"atsklinikum der privaten\\ Universit\"at Witten-Herdecke,\\
Campus K\"oln-Merheim, \\ 
Institut f\"ur Hygiene \\ 
Ostmerheimer Stra\ss e 200  \\ 
51109 K\"oln\\
Germany \\
\printead{e2}}
\end{aug}

\begin{abstract}
We consider the problem of comparing two diagnostic tests
based on  a sample of paired test results without true state determinations, 
in cases where the second test can reasonably be assumed to be at least as
specific as the first. For such cases, we provide two informative confidence
bounds: A lower one for the  prevalence times the sensitivity
gain of the second test with respect to the first,
and an  upper one for the sensitivity of the first test. 
Neither  conditional independence of the two tests nor 
perfectness of any of them needs to be assumed.

An application of the proposed confidence bounds to a sample of 256
pairs of laboratory test results for toxigenic Clostridium difficile
provides evidence for a dramatic sensitivity gain through 
first appropriately culturing Clostridium difficile from stool
samples before applying an enzyme-immuno-assay.  
\end{abstract}


\end{frontmatter}


\centerline{\em Dedicated to Abram M.~Kagan on the occasion of 
AMISTAT 2011 at Prague}

\section{Main results and applications} 
\subsection{Introduction and outline} \label{Subsec:Intro.outline}
Inference for sensitivities or specificities of diagnostic tests 
can be next to impossible if no suitable method for determining 
true states is available. Motivated by a real data problem described below, 
and in more detail 
in~\cite{Mattner.Winterfeld.Mattner.2009,Mattner.Winterfeld.Mattner.2010}, 
we consider here the situation where paired observations for two tests 
are given and where it can be assumed that the first test is less specific 
than the second. Can we then infer  from suitable observations 
that the second test is more sensitive, and hence better, than the 
first? And if yes, by how much? 

Theorem~\ref{Thm.main.result} in Subsection~\ref{Subsec.Main} 
below provides a simple and in some sense 
optimal answer. The necessary notation and concepts are carefully
explained before in Subsections \ref{Subsec.Math}-\ref{Subsec.Latent},
but some readers may wish  to start less formally
by  first consulting Subsection~\ref{Subsec.Example},
which introduces our motivating example, and then
proceed to the application of  Theorem~\ref{Thm.main.result}
given immediately after its statement. 
There it turns out that the answer to the above ``how much?'' question 
depends on upper bounds assumed  for the prevalence, but that 
nevertheless interesting upper bounds for the sensitivity of the 
first test can be given without such an assumption, using our 
Theorems~\ref{Thm.bound.Se1.new} and~\ref{Thm.bound.Se1}. 
Neither of our results uses any further assumptions, such as the 
conditional independence assumption as discussed
and criticized, for example,  in~\cite[Section~7.3]{Pepe.2003}.

We prove Theorems~\ref{Thm.main.result}-\ref{Thm.bound.Se1}
in the final Section~\ref{Sec.remaining.proofs},
after collecting auxiliary results on latent class models 
in Section~\ref{Sec.Aux.latent.class.models}
and proving them in Section~\ref{Sec.Proofs.Aux.latent.class.models}.

While there is a substantial literature on various aspects of the statistics of
diagnostic tests, see in particular the monographs
\cite{Abel.1993,Pepe.2003,Zhou.McClish.Obu.2002},
we are not aware of  a previous treatment of the problem 
considered here. Our  assumption that the first test is less specific than
the second may seem very  special, so let us point out 
that, for the purpose of obtaining upper bounds on the sensitivity of the first
test, our assumption  may  by Theorem~\ref{Thm.bound.Se1} 
replace the always less plausible assumption  of perfectness of the second 
test, see Subsection~\ref{Subsec:Example.CASS} for an example.

\subsection{Mathematical and probabilistic notation and conventions}\label{Subsec.Math}
We use ``iff'' as an abbreviation for ``if and only if''.
We write $\N:=\{1,2,3,\ldots\}$, $\N_0 :=\{0\}\cup\N$, 
and $\overline{\R}:=\R\cup\{-\infty,\infty\}$.
We put $x/0:=\infty$ for $x>0$, but we define $0/0$ below at each occurence
separately to be either $0$ or  $1/2$ or $1$. 
A subscript ``$+$'' indicates summation with respect to the variable
it replaces, as in $x_+ = \sum_{i=1}^n x_i$ for $x\in\R^n$
or in~\eqref{the.data} below  for $k\in\N_0^{\{0,1\}^2}$.
By contrast, a superscript
``$+$'' indicates the positive part, so  $x^+ = x\vee 0=\max\{x,0\}$ 
and correspondingly $x^- = (-x)\vee 0$
for $x\in\overline{\R}$.
As usual, the order theoretic operations $\wedge$ and $\vee$ are
computed first in expressions like $a\,b\wedge c := a\,(b\wedge c)= a\min\{b,c\}$.
 
If $\cX$ and $\cY$ are any sets, then 
\[
  \prob(\cX) &:=& \big\{ (\cX\ni x\mapsto p_x \in[0,1])\,:\,\sum_{x\in\cX}p_x=1\big\} \\
  \markov(\cX,\cY) &:=& 
\big\{ (\cX\times\cY \ni (x,y)\mapsto p_{y|x}) \,:\,p_{\cdot |x} \in\prob(\cY)
 \text{ for }x\in\cX\big\} 
\]
denote the set of all  discrete probability densities on $\cX$ and
the set of all discrete Markov transition densities from $\cX$ to $\cY$,
where  the standard dot notation 
$p_{\cdot |x}$ for
the partial function $y\mapsto p_{y|x}$ has been used. With 
$\mathrm{M}_{n,p}$ we denote the multinomial distribution with sample size
parameter $n$ and success probability vector $p\in\prob(\cX)$ for some $\cX$,
that is, $\mathrm{M}_{n,p}(\{k\}) = n!\prod_{x\in\cX}(p_x^{k_x}/k_x!)$
for $k\in\N_0^\cX$ with $\sum_{x\in\cX}k_x=n$.

\subsection{Confidence bounds and their comparison}\label{Subsec.Conf}
Let $\cP = (P_\theta : \theta \in\Theta )$ be 
a statistical model  on a sample space $\cX$ and let 
$\kappa : \Theta \rightarrow \overline{\R}$ be a parameter of interest.
We allow nonidentifiability of $\kappa$, that is, 
we may have $\theta_1,\theta_2 \in \Theta$ with $P_{\theta_1}= P_{\theta_2}$
but $\kappa(\theta_1)\neq \kappa(\theta_2)$.
For lack of any better name, let us call the pair
$(\cP,\kappa)$  an \defn{estimation problem}.
Let $\beta\in[0,1]$.  Then every measurable function 
$\underline{\kappa} : \cX \rightarrow\overline{\R}$ 
with $P_\theta( \underline{\kappa} \le \kappa(\theta))\ge \beta$
for every $\theta \in\Theta$ is called a \defn{lower $\beta$-confidence bound} for 
$(\cP,\kappa)$.

Now let $\underline{\kappa}$ and 
$\underset{\widetilde{}}{\kappa}$ be both lower $\beta$-confidence bounds
for $(\cP,\kappa)$.  Then everybody seems to agree that for preferring 
$\underline{\kappa}$ over $\underset{\widetilde{}}{\kappa}$, it would be 
desirable to have 
\la     \label{comp.conf.b.usual}
  P_\theta(  \underset{\widetilde{}}{\kappa}  \ge t) &\le&  
P_\theta( \underline{\kappa} \ge t)   
\qquad \text{ for }\theta\in\Theta \text{ and } t < \kappa(\theta)
\al
For example, Lehmann and Romano~\cite[page 72]{Lehmann.Romano.2005} 
would call $\underline{\kappa}$
\defn{uniformly most accurate} if  \eqref{comp.conf.b.usual} held 
for every $\underset{\widetilde{}}{\kappa}$ as above, 
but such a $\underline{\kappa}$ is known to exist
in exceptional cases only. 
The desideratum~\eqref{comp.conf.b.usual} could be supplemented
by conditions for $t\ge \kappa(\theta)$ in different 
ways, see~\cite[page 162]{Pfanzagl.1994} for one possibility, 
but we stick to~\eqref{comp.conf.b.usual} as it is.
Thus  we call  $\underset{\widetilde{}}{\kappa}$ \defn{worse}
than $\underline{\kappa}$, 
and equivalently $\underline{\kappa}$ \defn{better} 
than $\underset{\widetilde{}}{\kappa}$,
if \eqref{comp.conf.b.usual} holds,
and \defn{strictly} so,  
if in addition strict inequality holds 
in \eqref{comp.conf.b.usual} for at least one $\theta$ and one $t$. 
Accordingly, $\underline{\kappa}$ is called \defn{admissible}
as a $\beta$-confidence bound for $(\cP,\kappa)$,
if no other such bound $\underset{\widetilde{}}{\kappa}$ is strictly better. 
Finally,  $\underline{\kappa}$ and  $\underset{\widetilde{}}{\kappa}$
are called \defn{equivalent}, if each is worse than the other, that is,
if~\eqref{comp.conf.b.usual} holds with ``$=$'' in place of ``$\le$''.  

\subsection{Latent class models for diagnostic tests}\label{Subsec.Latent}
Informally speaking, a (dichotomous) diagnostic test is a procedure yielding a guess
$\in$ $\{0,1\}$  for the state $\in$ $\{0,1\}$ of  any item 
belonging to some specified population. In this context,
$0$ is called {\em negative} and $1$ is called {\em positive}. 
In medicine, the population often consists of persons, 
for whom a positive state means actually having a certain disease,
and a positive diagnosis means to be guessed to have the disease.   
The accuracy of a diagnostic test is modelled by two numbers
called {\em specificity} and {\em sensitivity}, 
with specificity interpreted as  
the probability that a random negative item is diagnosed as negative,
and sensitivity as  
the probability that a random positive item is diagnosed as positive.
The probability of  diagnosing a random item from the whole population 
as positive, say, then of course depends also on the {\em prevalence},
which is the probability of such an item to be actually positive.
If we formalize the above, for samples of size~$n$ rather than~$1$,
and also admitting more generally $d$~tests, rather than just one, 
to be applied to every item, we arrive at  the following 
model considered in essence already in \cite{Gart.Buck.1966}.

\label{sec.latent.class.model}
Let  $d\in\N$ and
\[
 \Theta_d &:=& \prob(\{0,1\})\times \markov(\{0,1\}, \{0,1\}^d)
\]
For $\theta = (\pi,\chi) \in \Theta_d$, let 
$\mu(\theta)\in \prob(\{0,1\}^d )$ denote the 
second marginal of  the density 
\la         \label{eq.def.joint.dens}
 \{0,1\}\times \{0,1\}^d \ni (i,j) \mapsto \pi^{}_i\chi^{}_{j|i}
\al
so that 
\la    \label{eq.def.lambda}\label{eq.def.mu}
 \big(\mu(\theta)\big)_{j}
 &=&  \sum_{i=0}^1
 \pi^{}_i\chi^{}_{j|i}
 \quad\text{ for }j\in\{0,1\}^d \text{ and } \theta\in\Theta_d
\al
Finally, with a given $n\in\N$ often
notationally surpressed in what follows, let   
\la          \label{Eq:Def.P.theta}
  P_\theta &:=& \mathrm{M}_{n,\mu(\theta)} \qquad\text{ for }\theta\in\Theta_d
\al
\begin{defin}
Let $d,n\in \N$. The \textit{\em (full) latent class model} for a sample of 
size $n$ of combined results of $d$ diagnostic tests with unknown 
characteristics and for a state with unknown prevalence is
$ \cP_d:= ( P_\theta  : \theta \in\Theta_d)$.
\end{defin}
The interpretation of the parameter $\theta =(\pi,\chi)$ in 
this model is as follows: $\pi_1$ is the prevalence of positive states
and $\chi$ is the joint characteristics of the $d$~diagnostic tests.

For example, let $d=2$. Then $\chi^{}_{01|0}$ is the probability 
that a random negative (see the last bit of the subscript) 
is diagnosed negative by the first test (see the first bit of the subscript)
and positive by the second (see the second bit of the subscript).
And $\chi^{}_{0+|0}=\chi^{}_{00|0} +\chi^{}_{01|0}$
is then accordingly the probability that a random negative
is diagnosed negative by the first test, that is, the 
specificity of  the first test.
More systematically, and introducing a notation used below,
we put
\la          \label{Eq.Def.partial.chi}
\chi^{(1)}_{\iota|i} := \chi^{}_{\iota+|i}
&\text{ and }&\chi^{(2)}_{\iota|i} := \chi^{}_{+\iota|i}
\qquad\text{ for }i,\iota\in\{0,1\}
\al
and regard $\chi^{(1)}, \chi^{(2)} \in \markov( \{0,1\},\{0,1\})$
as the characteristics of the first and of the second test, respectively.

Coming back to general $d$, formula~\eqref{eq.def.joint.dens}
gives the joint density of  a true state determination together with
the results of the $d$ tests, for an item picked  at random from the whole 
population, and $\mu(\theta)$ is the marginal density 
corresponding to unobservability of the true state. 
Finally, the multinomial distribution
$P_\theta=\mathrm{M}_{n,\mu(\theta)}$ models 
testing thus a random sample of size~$n$ from the (conceptually infinite) population,
and  counting  just the number of occurences of each possible combination of the $d$ test results.

In this paper, motivated by the application sketched in
Subsection~\ref{Subsec.Example} below,
we are mainly interested  in  the case of $d=2$, and 
here  in particular in the submodel assuming that the specificity 
of the first test is at most equal to the specificity of the second.
In terms of the parameter $\theta=(\pi,\chi)\in\Theta_2$ 
and with the notation introduced in~\eqref{Eq.Def.partial.chi} above, 
this assumption is expressed as  $\chi^{(1)}_{0|0}\le \chi^{(2)}_{0|0}$.

\begin{defin}  \label{Def:restr.lat.cl.mod}
Let $n\in \N$. In this paper,
the \textit{\em restricted latent class model} for a sample of 
size $n$ of combined results of two diagnostic tests with unknown 
characteristics and for a state with unknown prevalence is
$\cP_{2,\le}:= ( P_\theta  : \theta \in\Theta_{2,\le})$
with $\Theta_{2,\le} := \{(\pi,\chi) \in\Theta_2 :  \chi^{(1)}_{0|0}\le \chi^{(2)}_{0|0}   \}$.
\end{defin}

\subsection{Example: A comparison of two tests for diagnosing toxigenic
Clostridium difficile}   \label{Subsec.Example}
{\em Clostridium difficile} is a certain species of bacteria. Some of these, 
called {\em toxigenic}, have the potential to produce one
or both of certain toxins, called {\em A} and {\em B}.
Toxigenic Clostridium difficile is responsible for one of the most 
prevalent infections of the human gut. 
It may lead to severe courses of infection and is easily transmitted in 
hospitals. A fast  and accurate diagnosis would be highly desirable
for initiating adequate therapy and preventing transmissions to other 
patients. Unfortunately, so far no diagnostic test, not even a  complex and 
time-consuming  one, has been proven to be highly accurate, that is, 
with specificity and sensitivity close to 1.

Available diagnostic tests are applied to stool specimens of patients
with diarrhoea, using one of the following three methods, with details
to be specified.  The first, simple and a matter of a few hours, 
consists in performing an enzyme-immuno-assay~({\em EIA}) 
for the direct detection of toxin A or B in the stool specimen. 
The second, taking about 3~days,  consists in trying to culture 
Clostridium difficile (possibly nontoxigenic) from the stool specimen
on an appropriate medium 
and applying then a ``confirmatory test'' for toxin A or B,
for example an EIA as above, to any cultured colonies.
The third, again taking about 3 days,
tests the cytotoxicital potential of the stool specimen
by applying it to a  vero-cell culture 
({\em cytotoxicity neutralisation test}). For several such tests, different
accuracy values were published during the last years,
often obtained by assuming the cytotoxicity neutralisation test
to be a sufficiently accurate reference test or ``gold standard'',
see~\cite{Mattner.Winterfeld.Mattner.2009,Mattner.Winterfeld.Mattner.2010}
for appropriate references.

One goal 
of~\cite{Mattner.Winterfeld.Mattner.2009,Mattner.Winterfeld.Mattner.2010}
was to compare 
a test according to the first method
described above ({\em Test 1} or {\em direct test})
with a test according to the second method, with the confirmatory test
being the same EIA as in the direct test ({\em Test 2} or {\em culture test}).
Both tests were applied to each stool specimen of a sample of size~$256$,
consisting of all liquid specimens sent to a microbiological laboratory
during two consecutive months. 
The observed data were
\la \label{the.data}
 \begin{array}{rr|r}
  k^{}_{00}=          210 &\quad k^{}_{01}=\phantom{2}20& k^{}_{0+}=230 \\
  k^{}_{10}=\phantom{21}4 &\quad k^{}_{11}=\phantom{2}22& k^{}_{1+}=\phantom{2}26 \\ \hline
  k^{}_{+0}= 214          &\quad k^{}_{+1}=\phantom{2}42& k^{}_{++}=256
 \end{array}
\al
where, for example,  $k^{}_{01}$ is the number of specimens tested negative with Test~1
and positive with Test~2. 
True states were unobservable. The prevalence of toxigenic  Clostridium
difficile, in the population of all liquid stool samples sent to 
a laboratory for microbiological investigation,
is certainly not known precisely, but is believed to be very roughly $15\%$.  
So far it seems  natural to use the full latent class  
model $\cP_2$ for analyzing the data.
However, as  the EIA is applied in Test~1 to the whole stool specimen 
and in Test~2 only to a part of a culture from the specimen 
already identified as Clostridium difficile, it seems very plausible
to assume that Test~2 is at least as specific as Test~1.
This  suggests 
that the  restricted latent class model $\cP_{2,\le}$ could be used, 
and that then
the superiority of Test~2 would follow if the latter can be proved 
to be also more  sensitive than Test~1.
Theorem~\ref{Thm.main.result} in the next section 
is formulated with a view towards situations like the present,
taking into account both models,  $\cP_2$ and  $\cP_{2,\le}$.

\subsection{Main results. Application to  the comparison of tests
for diagnosing toxigenic Clostridium difficile} \label{Subsec.Main}
\begin{thm}\label{Thm.main.result} 
Let $\beta\in[0,1]$, $n\in \N$, 
\la      \label{Deltaunderline}
  \underline{\Delta}: \big\{k\in\N_0^{\{0,1\}^2} : k_{++} = n\big\} \rightarrow [-1,1]
\al
be a function,  
and   $ \cM := (\mathrm{M}_{n,q} : q \in \prob(\{0,1\}^2))$
be a quadrinomial model.

{\sc A.}
The following three assertions are equivalent:

(i) $ \underline{\Delta}$ is a lower $\beta$-confidence bound in the model 
$\cM$ and for the parameter 
\la          \label{Param.interest.multinomial}
 q&\mapsto& q^{}_{01}-q^{}_{10}
\al 

(ii) $\underline{\Delta}$ is a lower $\beta$-confidence bound 
in the full latent class model $\cP_2$ and for  the parameter 
\la                                  \label{Param.interest.full}
  (\pi,\chi) &\mapsto&  
  \pi_1\, \big(\,\chi^{(2)}_{1|1}-\chi^{(1)}_{1|1}\,\big)
  -  (1-\pi_1)\, \big(\,\chi^{(2)}_{0|0}- \chi^{(1)}_{0|0}\,\big)
\al
   
(iii) $\underline{\Delta}$ is a lower $\beta$-confidence bound
in the restricted latent class model $\cP_{2,\le}$ and  for  the parameter 
\la                                   \label{Param.interest.restr}
  (\pi,\chi) &\mapsto& \pi_1\, \big(\,\chi^{(2)}_{1|1}- \chi^{(1)}_{1|1}\,\big)
\al
 
{\sc B.} 
Let $\underline{\Delta}$  obey 
the above conditions (i)-(iii) and let
$\underset{\widetilde{}}{\Delta}$ be another such function.
Then $\underset{\widetilde{}}{\Delta}$ is worse than $\underline{\Delta}$
as a lower $\beta$-confidence bound for 
$(\cM, \eqref{Param.interest.multinomial})$
iff it is so for $(\cP_{2}, \eqref{Param.interest.full})$,
and if it is so for  $(\cP_{2,\le}, \eqref{Param.interest.restr})$.
(Once ``iff'', once ``if''.)

{\sc C.} If $\underline{\Delta}$ is admissible 
as a $\beta$-confidence bound for one of the 
problems $(\cM, \eqref{Param.interest.multinomial})$ and 
$(\cP_{2}, \eqref{Param.interest.full})$,
then so it is for the other 
and for $(\cP_{2,\le}, \eqref{Param.interest.restr})$.
\end{thm}

See Section~\ref{Sec.remaining.proofs} for a proof of this and the other two 
theorems of this subsection.
We proceed to illustrate Theorem~\ref{Thm.main.result} 
by its application to the example from Subsection~\ref{Subsec.Example}.
Let $\beta\in[0,1]$ and $n\in \N$ be fixed. 
Wanted is a ``good''
confidence bound $\underline{\Delta}$ as in~\eqref{Deltaunderline}
and (iii) above.
Parts B and C Theorem~\ref{Thm.main.result}  suggest choosing $\underline{\Delta}$
to be a ``good'' confidence bound as in~(i).
We put
\[
  \underline{\Delta}(k) &:=& \ell(k_{01},k_{10},k_{00}+k_{11} ) \qquad\text{
    for }k\in\N_0^{\{0,1\}^2} \text{ with } k_{++} = n 
\]
where 
$\ell :  \{k\in\N_0^3\,:\, k_+ = n \} \rightarrow [-1,1]$
is the Lloyd-Moldovan lower $\beta$-confidence bound 
for the coordinate difference $\prob(\{1,2,3\}) \ni p\mapsto p_1-p_2$
in the trinomial model $\big(\mathrm{M}_{n,p} : p \in \prob(\{1,2,3\})\big)$,
see Subsection~\ref{Subs.Lloyd.Moldovan}. Then $\underline{\Delta}$
satisfies~\eqref{Deltaunderline} and~(i). 
With  $\beta = 0.95$ and  the data~$k$
from~\eqref{the.data}, we get  
\[
 \underline{\Delta}(k)&=& \ell(20,4,232) \,\,\, = \,\,\, 
 \mathtt{0.0320}
\]
as our lower confidence bound in (i),  corresponding to the point 
estimate $\frac{20}{256}-\frac{4}{256}=\frac{1}{16}= 0.0625$.
(Here and below, numbers in typescript like  $\mathtt{0.0320}$
are rounded consistently with the inequalities claimed.)
Thus,   assuming the restricted latent class model $\cP_{2,\le}$ and
using~(iii), we get the confidence statement 
\la      \label{Diff.significant}
  \chi^{(2)}_{1|1}- \chi^{(1)}_{1|1} &\ge& \frac{\mathtt{0.0320}}{\pi_1} 
\al
with $\pi_1>0$, so that Test~2 is significantly more sensitive
than Test~1 and hence, being at least as specific by assumption, 
significantly better. Without any upper bound on the 
prevalence $\pi_1$, the best lower bound for  the sensitivity gain  $\chi^{(2)}_{1|1}-
\chi^{(1)}_{1|1}$ of the culture test with respect to the direct test
we can obtain from~\eqref{Diff.significant}
is $\mathtt{0.0320}$.
But assuming   some plausible upper bound implies 
a dramatic sensitivity gain; for example, $\pi_1\le 0.15$
yields  $\chi^{(2)}_{1|1}- \chi^{(1)}_{1|1} \ge \mathtt{0.0320/0.15} 
= \mathtt{0.21}$.
This would imply in particular
$ \chi^{(1)}_{1|1} \le 1- (\chi^{(2)}_{1|1}- \chi^{(1)}_{1|1}) 
\le 1- \mathtt{0.21} = \mathtt{0.79}$
and hence a very poor sensitivity of the direct test.
It is remarkable that the latter conclusion, with a slightly larger 
bound, can be obtained
without any assumption on the prevalence by  using the following 
theorems, see \eqref{Bound.Se.direct.1} and \eqref{Bound.Se.direct.2}
below.

\begin{thm}\label{Thm.bound.Se1.new} 
Let $\beta\in[0,1]$, $n\in \N$, 
\la      \label{Eq:S.overline}
  \overline{S}: \big\{k\in\N_0^{\{0,1\}^2} : k_{++} = n\big\} \rightarrow [0,1]
\al
be a function,  
and   $ \cM := (\mathrm{M}_{n,q} : q \in \prob(\{0,1\}^2))$
be a quadrinomial model.

{\sc A.} $ \overline{S}$ is 
an upper $\beta$-confidence bound in the model 
$\cM$ and for  the parameter 
\la          \label{Param.interest.multinomial.2}
 q&\mapsto& 
 \frac{q^{}_{1+}}{ q^{}_{1+} + q^{}_{01}} 
 \vee\left(\frac{q^{}_{11}}{\left( q^{}_{+1}-q^{}_{10}\right)^+}\wedge 1\right)
 \quad\text{ with } \quad\frac00 :=1
\al
iff it is so in the  restricted latent class model $\cP_{2,\le}$
and for the parameter 
\la                                  \label{Param.interest.restr.2}
  (\pi,\chi) &\mapsto&      \chi^{(2)}_{1|1} 
\al

{\sc B.} Let $\overline{S}$ obey the equivalent conditions from part A, and
let $\widetilde{S}$ be another such function. If  $\widetilde{S}$ is worse
than $\overline{S}$ as an upper $\beta$-confidence bound for 
$(\cP_{2,\le},\eqref{Param.interest.restr.2})$, then so it is for  
$(\cM,\eqref{Param.interest.multinomial.2})$. 

{\sc C.} If $ \overline{S}$ is admissible as a $\beta$-confidence bound for 
$(\cM, \eqref{Param.interest.multinomial.2})$, then so it is for 
 $(\cP_{2,\le}, \eqref{Param.interest.restr.2})$.
\end{thm}

We get a confidence bound for $(\cM,\eqref{Param.interest.multinomial.2})$,
as needed for applying 
Theorem~\ref{Thm.bound.Se1.new}~A, from confidence bounds in 
certain trinomial models, similarly to but slightly less obviously than for the
situation of  Theorem~\ref{Thm.main.result}: 

\begin{thm}\label{Thm.bound.Se1} 
Let $\beta\in[0,1]$, $n\in \N$, and 
$u: \{k\in \N_0^3 : k_+ \le n \}  \rightarrow [0,\infty]$ be
a function such that, for every $m\in\{0,\ldots,n\}$,
the  restriction of $u$ to     $\{k\in \N_0^3 : k_+ = m \}$ 
is an upper $\beta$-confidence bound
in the trinomial model $\big(\mathrm{M}_{m,p} : p \in \prob(\{1,2,3\})\big)$
and  for the parameter
\la               \label{Eq:Second.trinomial.parameter} 
 p& \mapsto&       
  \left(1-p_2\right) \vee 
 \left(\,\frac{1-p_1-p_2}{\left(1-2\,p_1\right)_{}^+}   \wedge 1 \,  \right)
 \quad\text{ with } \quad\frac 00 := 1
\al 
Then the function
\la      \label{u.without.k00}
  \big\{k\in\N_0^{\{0,1\}^2} : k_{++} = n\big\}\,\,\,\ni\,\,\,        
   k&\mapsto& u(k_{10}, k_{01}, k_{11})
\al
is an upper $\beta$-confidence bound for 
$(\cP_{2,\le},\eqref{Param.interest.restr.2})$.
\end{thm}

As we are not aware of a function $u$ as assumed in Theorem~\ref{Thm.bound.Se1}
and also  well-founded and easily available for practical computation,  
we use here the following ad hoc method:  
Let $u_0$ denote the Lloyd-Moldovan upper $\beta$-confidence bound
corresponding to the lower bound $\ell$ used above. Then, since 
\[
  \textrm{R.H.S.\eqref{Eq:Second.trinomial.parameter}}
 &=&   \left(1-p_2\right) \vee 
 \left(\,\frac{1-2\,p_1 + p_1 -p_2}{\left(1-2\,p_1\right)_{}^+}   \wedge 1 \,  \right)
 \quad\text{ with } \quad\frac 00 := 1 \\
 &\le & \left( 1+p_1-p_2 \right) \wedge 1
\] 
for $p\in\prob(\{1,2,3\})$, we may take $u:= \left(1+u_0\right)\wedge 1$ in 
Theorem~\ref{Thm.bound.Se1}. Applied to our data~\eqref{the.data},
this yields 
$u(k_{10}, k_{01}, k_{11})= \left( 1+ u_0(4,20,22) \right)\wedge 1 
= \mathtt{0.83}$  
and thus 
\la       \label{Bound.Se.direct.1}
  \chi^{(1)}_{1|1} &\le &  \mathtt{0.83}
\al
with confidence $0.95$, in the restricted latent class model without further
assumptions.

Going back to~\eqref{Diff.significant}, obtained under the restricted latent
class model, Part~A of Theorem~\ref{Thm.main.result} suggests that 
we should perhaps rather state 
\la      \label{Diff.significant.2}
  \chi^{(2)}_{1|1}- \chi^{(1)}_{1|1} &\ge& \frac{\mathtt{0.0320}}{\pi_1} 
  \,+\, \frac {1-\pi_1}{\pi_1}\, \big(\,\chi^{(2)}_{0|0}- \chi^{(1)}_{0|0}\,\big) 
\al
as a valid confidence statement under the full latent class model.
This not only makes obvious the effect of the possibility 
$ \chi^{(2)}_{0|0}- \chi^{(1)}_{0|0} <0$ in the larger model,
drastically decreasing the lower bound for the sensitivity difference,
but also the possibly drastic increase if we  actually  have 
$ \chi^{(2)}_{0|0}- \chi^{(1)}_{0|0} >0$ and $\pi_1$ rather small.

So far, we have for simplicity only considered part of 
the data
from~\cite{Mattner.Winterfeld.Mattner.2009,Mattner.Winterfeld.Mattner.2010}.
There, we actually applied  the direct test and three versions of the 
culture test, differing in the culture media used, 
to each of the 256~specimens. The media are called I, II, III 
in~\cite{Mattner.Winterfeld.Mattner.2009,Mattner.Winterfeld.Mattner.2010},
and here~\eqref{the.data} presents just the results for the direct test
and for the culture test with  medium~II.  
Bounds analogous to 
the above lower confidence bound for the sensitivity gain through 
culturing with medium II, with the exemplary assumption $\pi_1\le 0.15$,
were computed for media I and III,
resulting in $\mathtt{-0.04}$ for I (so no statistically significant gain here)
and $\mathtt{0.02}$ for III.
For obtaining the upper confidence bound on the sensitivity of the direct
test, without any assumption on the prevalence,
we compared  in~\cite{Mattner.Winterfeld.Mattner.2010} 
the direct test with the logical 
\verb|or|ing of the three culture tests,
which diagnoses a specimen as positive if at least one of the three does so,
yielding the data $k_{00}= 209$, $k_{01}=21$, $k_{10}=4$, $k_{11}=22$ 
rather than~\eqref{the.data}, and hence the confidence statement 
\la       \label{Bound.Se.direct.2}
 \chi^{(1)}_{1|1}&\le& \left( 1+ u_0(4,21,22) \right)\wedge 1 
 \,\,\,=\,\,\, \mathtt{0.81} 
\al

\subsection{The Lloyd-Moldovan confidence bound for a coordinate 
difference of a multinomial parameter} \label{Subs.Lloyd.Moldovan}
The best currently  available confidence bound $\ell$ 
as needed in Theorem~\ref{Thm.main.result} 
appears  to be the one proposed and implemented
by Lloyd and Moldovan~\cite{Lloyd.Moldovan.2007}:
To compute it, load their program into {\tt R} with  
\verb|load("sm_file_SIM2708_2")|, 
type \verb|bcl(cl.side=-1)|, 
where ``\verb|-1|'' asks for the lower rather than
the default upper bound obtainable with just \verb|bcl()|,
enter the three numbers $\mathtt{x}=k_1$, $\mathtt{t}=k_1+k_2$ 
und $\mathtt{n}=k_1+k_2+k_3$, with return after each, 
and then a few more returns,
assuming here $\beta=0.95$ for simplicity.

\subsection{Example:
Robust upper confidence bounds for the sensitivities
of diagnostic tests for coronary artery disease} 
\label{Subsec:Example.CASS}
This subsection uses part of a standard dataset,
given in~\cite[Table 5]{Leisenring.Alonzo.Pepe.2000} 
and~\cite[pp.~8, 17, 22]{Pepe.2003}  and drawn 
from~\cite{Weinert.et.al.1979},
to exemplify the final sentence of Subsection~\ref{Subsec:Intro.outline}.
We consider evaluating two diagnostic tests for
coronary artery disease~({\em CAD}). 
This disease is the most frequent cause of 
myocardic infarction, which in turn is the most frequent 
cause of death in developed countries. 

The first test considered is a dichotomized exercise 
stress test~({\em EST}\,), the second a  dichotomized chest pain 
history~({\em CPH}\,). These two tests and a dichotomized 
arteriography~({\em A}) were performed on each of 1465 men. The dataset is
a three-way table of counts $k\in\N_0^{\{0,1\}^3}$ with 
$k_{+++}=1465$ and with the indexing here corresponding to the ordering 
EST, CPH, A: 
$k_{000}=151$ men negative for all three tests,
$k_{001}= 25$ positive only for A,
$k_{010}=176$ positive only for CPH,
$k_{011}=183$, 
$k_{100}= 46$  positive only for EST,
$k_{101}= 29$,
$k_{110}= 69$,
$k_{111}=786$. 
As usual, it is assumed that the 1465 trivariate observables are independent
and identically distributed. Let $ k^{\mathrm{EST}} 
:= (k_{i+j} : (i,j)\in\{0,1\}^2)$ denote the marginal table
for just the results of EST and A, and let analogously
$ k^{\mathrm{CPH}} := (k_{+ij} : (i,j)\in\{0,1\}^2)$ 
be the marginal table for CPH and A. Thus 
\[
 \begin{array}{rr}
  k^{\mathrm{EST}}_{00} = 327 &\quad  k^{\mathrm{EST}}_{01}= 208 \\
  k^{\mathrm{EST}}_{10}= 115 &\quad  k^{\mathrm{EST}}_{11}= 815  
 \end{array}
 &\qquad&
\begin{array}{rr}
  k^{\mathrm{CPH}}_{00}= 197&\quad k^{\mathrm{CPH}}_{01}= \phantom{9}54 \\ 
 k^{\mathrm{CPH}}_{10}= 245&\quad k^{\mathrm{CPH}}_{11}= 969   %
\end{array}
\]
If, as in~\cite{Leisenring.Alonzo.Pepe.2000,Pepe.2003}, the 
test~A is assumed to be perfect, then  we  get the following 
four separate $95\%$ binomial confidence statements (ignoring corrections 
for quadruplicity) for the sensitivities $\Se^{\mathrm{EST}}$ and   
$\Se^{\mathrm{CPH}}$ and the specificities $\Sp^{\mathrm{EST}}$ 
and $\Sp^{\mathrm{CPH}}$ of the tests EST and CPH,
\la    \label{Eq:CASS.sensitivities}
 \mathtt{0.770} \le \Se^{\mathrm{EST}} \le \mathtt{0.821}&\qquad&
 \mathtt{0.931} \le \Se^{\mathrm{CPH}} \le  \mathtt{0.961}  \\
 \mathtt{0.696} \le \Sp^{\mathrm{EST}} \le  \mathtt{0.781} &\qquad&
\mathtt{0.398} \le \Sp^{\mathrm{CPH}}  \le \mathtt{0.494}  \nonumber
\al
using, for example, the {\tt R}-command \verb|binom.test(c(815,208))| 
for the first interval, and we may conclude 
that neither  EST nor CPH is sufficiently accurate.
The perfectness of A means that its specificity $\Sp^{\mathrm{A}}$
and  its sensitivity $\Se^{\mathrm{A}}$ are both equal to $1$,
or rather practically very nearly so.
Here the assumption  $\Sp^{\mathrm{A}}=1$ appears quite reasonable 
from the medical point of view, but  $\Se^{\mathrm{A}}=1$ does not.
Using now only the weaker assumption  
$\Sp^{\mathrm{EST}}\le\Sp^{\mathrm{A}}$ or $\Sp^{\mathrm{CPH}}\le\Sp^{\mathrm{A}}$,
respectively, we get the two separate $95\%$ upper confidence bound statements
\[
  \Se^{\mathrm{EST}} \le   \mathtt{0.945}
 &\qquad& \Se^{\mathrm{CPH}} \le 1
\]
using Theorem~\ref{Thm.bound.Se1} as in Subsection~\ref{Subsec.Main},
computing $u(115,208,815)$ and $u(245,54,969)$ with the ad hoc function 
$u$ indicated there. The second bound is unfortunately trivial, 
but the first,  while of course weaker than the 
statement from~\eqref{Eq:CASS.sensitivities} obtained under a much stronger
assumption, is still good enough to show that EST is far from perfect:
EST fails to diagnose CAD for at least every twentieth CAD patient.

\section{Auxiliary results on latent class  models}
\label{Sec.Aux.latent.class.models}
In this section we describe images,  under various  parameters of interest,
of the preimage 
$\mu^{-1}(\{q\})$ $=$ $\{\theta =(\pi,\chi)\in \Theta_d : \mu(\theta)=q\}$
in Subsections~\ref{Subs.Aux.LCM.1} and~\ref{Subs.Aux.LCM.2},
and of a similar preimage with $\Theta_{2,\le}$ in place of 
$\Theta_d$
in Subsection~\ref{Subs.Aux.LCM.3}, 
of a given $q \in\prob(\{0,1\}^d)$ under the function
$\mu$ defined in~\eqref{eq.def.mu}.
Informally speaking,  this amounts to determining 
the exact joint range of the possible values of the
prevalence, sensitivities, and specificities 
(Lem\-ma~\ref{Lem.d=1} for $d=1$ and 
Lem\-ma~\ref{Basic.lemma} for $d=2$), or certain functions of these  
(Lemmas~\ref{Lemma.1.3}-\ref{Lemma.2.16}),
assuming the density  $q$ of the joint test results 
as known or, in a more practical interpretation,
estimated with high accuracy 
from a very large sample of joint test results. 
For example, using here, for the purpose of illustration  only, 
$q = \hat{q}:=\frac k{k_{++}}$ based on the data $k$ from~\eqref{the.data},
that is 
\[
 \begin{array}{cl|l}
  \hat{q}^{}_{00}= 0.820&\quad  \hat{q}^{}_{01}= 0.078&  \hat{q}^{}_{0+}=0.90\\
  \hat{q}^{}_{10}= 0.016&\quad  \hat{q}^{}_{11}= 0.083&  \hat{q}^{}_{1+}=0.10\\\hline
  \hat{q}^{}_{+0}=0.84\phantom{0}&\quad  \hat{q}^{}_{+1}=0.16\phantom{0}&  \hat{q}^{}_{++}=1
 \end{array}
\]
the pictures  of $C$ and  $C_\le$ displayed below near the corresponding 
Lemmas~\ref{Lemma.3.4.new} and~\ref{Lemma.1.5} 
show as hatched regions
the exact joint ranges of the possible values of the
prevalence and the sensitivity difference, the first in the full latent
class model, and the second in the restricted  one.

All these lemmas, needed to prove 
Theorems~\ref{Thm.main.result} and~\ref{Thm.bound.Se1}
in Section~\ref{Sec.remaining.proofs} below,
are proved in Section~\ref{Sec.Proofs.Aux.latent.class.models},
where the less interesting results of
Subsections~\ref{Subs.Aux.LCM.1} and~\ref{Subs.Aux.LCM.2}
are used for obtaining 
the more important results of  Subsection~\ref{Subs.Aux.LCM.3}.

We have found it suggestive to denote  below certain ``variables''
with $\Pr$, $\Sp$, $\Se$, $\Sp_1$, $\Sp_2$, $\Se_1$,  $\Se_2$, and $\Delta\Se$. 
Perhaps it should be pointed out that, for example,
denoting a variable by  $\Se_1$ in Lem\-ma~\ref{Lemma.2.4}
does not imply  that~$\Se_1$ be the first coordinate of 
some tuple 
called $\Se$. This differs from our use of subscripts
for $\pi$ and $\chi$, for example in   the definition of $A$ in 
Lemma~\ref{Basic.lemma}, where $\pi_1$ is understood to be the 
last coordinate of $\pi=(\pi_0,\pi_1)$.

\subsection{The case $d=1$ and a partial reduction to it}\label{Subs.Aux.LCM.1}
In this subsection,  we write more precisely $\mu_d$ for the function $\mu$
from \eqref{eq.def.mu}. 
\begin{lem}  \label{Lem.d=1}
If $q\in \prob(\{0,1\})$, then
$
  \{(\pi_1,\chi^{}_{0|0},\chi^{}_{1|1})  : (\pi,\chi)\in\mu_1^{-1}(\{q \}) \}$
$=$ $\big\{\, (\Pr,\Sp, \Se)\in[0,1]^3 \,:\,
 (1-\mathrm{Pr} )\,  (1-\mathrm{Sp}) + \mathrm{Pr} \,\mathrm{Se} = q^{}_1\,\big\}$.
\end{lem}
We recall the dot notation for functions explained in Subsection~\ref{Subsec.Math}.
\begin{lem} \label{Lem.d=2.d=1}
Let $q\in \prob(\{0,1\}^2)$. Then
\la
  \{(\pi,\chi^{(1)}) : (\pi,\chi)\in\mu_2^{-1}(\{q \}) \}
 &=& \mu_1^{-1}(\{q_{\cdot+}\}) \label{eq.mu2mu1.1} \\
 \{(\pi,\chi^{(2)}) : (\pi,\chi)\in\mu_2^{-1}(\{q \}) \} \label{eq.mu2mu1.2}
 &=& \mu_1^{-1}(\{q_{+\cdot}\})
\al
\end{lem}
\subsection{The case $d=2$ for the full latent class model}\label{Subs.Aux.LCM.2}
In this subsection and  in the next one, 
we return to the shorter notation $\mu$ instead of $\mu_2$, and we assume that 
$q\in\prob(\{0,1\}^2)$ is fixed. 
\begin{lem} \label{Basic.lemma}
$ A := \big\{\,\left(\pi^{}_1,\chi^{(1)}_{0|0},\chi^{(1)}_{1|1},\chi^{(2)}_{0|0},
 \chi^{(2)}_{1|1}\right) \, :\,   (\pi,\chi)\in \mu^{-1}(\{q\})  \big\} $
is the nonempty
set of all $(\Pr,\Sp_1,\Se_1,\Sp_2,\Se_2) \in [0,1]^5$ 
satisfying  the relations 
\la 
 (1-\Pr )\,  (1-\Sp_1) + \Pr \,\Se_1 &=& q^{}_{1+} \label{eq.cond1}\\ 
(1-\Pr )\,  (1-\Sp_2) + \Pr \,\Se_2  &=& q^{}_{+1} \label{eq.cond2}\\
 (1-\Pr )\,\Sp_1 \wedge \Sp_2   + \Pr\,(1-\Se_1 \vee \Se_2)
  &\ge& q^{}_{00} \label{eq.cond3}  \\
  (1-\Pr )\,(\Sp_1+\Sp_2-1)_{}^+ + \Pr\,(1-\Se_1-\Se_2)_{}^+
  &\le&  q^{}_{00}  \label{eq.cond4}
\al
or, equivalently, 
\la
 \Pr\,(\Se_2-\Se_1)&=&(1-\Pr)(\Sp_2-\Sp_1)+q^{}_{01}-q^{}_{10}\label{eq.cond5}\\
 \Pr\,(\Se_1+\Se_2-1)&=&(1- \Pr)(\Sp_1+\Sp_2-1)+q^{}_{11}-q^{}_{00}\label{eq.cond6}\\
 -q^{}_{10} \,\,\le\,\,\Pr\,(\Se_2-\Se_1)&\le& q^{}_{01}\label{eq.cond7}\\ 
\qquad -q^{}_{00}\,\,\le\,\,\Pr\,(\Se_1+\Se_2-1) &\le & q^{}_{11}\label{eq.cond8}
\al
\end{lem}

\begin{lem}\label{Lemma.1.3}
$ B :=
 \big\{\,\left(\pi^{}_1,\chi^{(1)}_{1|1},\chi^{(2)}_{1|1}  \right)\, :\,
   (\pi,\chi)\in \mu^{-1}(\{q\})
 \big\} $
is the  nonempty set of all 
$(\Pr,\Se_1,\Se_2) \in [0,1]^3$ satisfying  the relations 
\eqref{eq.cond7}, \eqref{eq.cond8}
and 
\la
  \Pr - q^{}_{0+}&\le&\Pr\, \Se_1 \,\,\le\,\,q^{}_{1+} \label{eq.cond10} \\
  \Pr - q^{}_{+0}&\le&\Pr\, \Se_2 \,\,\le\,\,q^{}_{+1} \label{eq.cond11} 
\al 
\end{lem}

\begin{lem}\label{Lemma.3.4.new}
$ C :=\big\{\,\left(\pi^{}_1, \chi^{(2)}_{1|1} -   \chi^{(1)}_{1|1}\right)\, :\,
  (\pi,\chi)\in \mu^{-1}(\{q\})  \big\} $
is the  nonempty set of all $(\Pr,\Delta\Se)\in[0,1]\times[-1,1]$ satisfying  
the  inequalities
\la        \label{eq.cond14}
  (-q^{}_{10})\vee (q^{}_{01}-q^{}_{10} +\Pr-1)  
  &\le& \Pr\,\Delta\Se  \\  
  &\le&  \nonumber
  q^{}_{01}\wedge (q^{}_{01}-q^{}_{10} + 1-\Pr )
\al
\end{lem}
\includegraphics[width=\linewidth,
  keepaspectratio
  ]{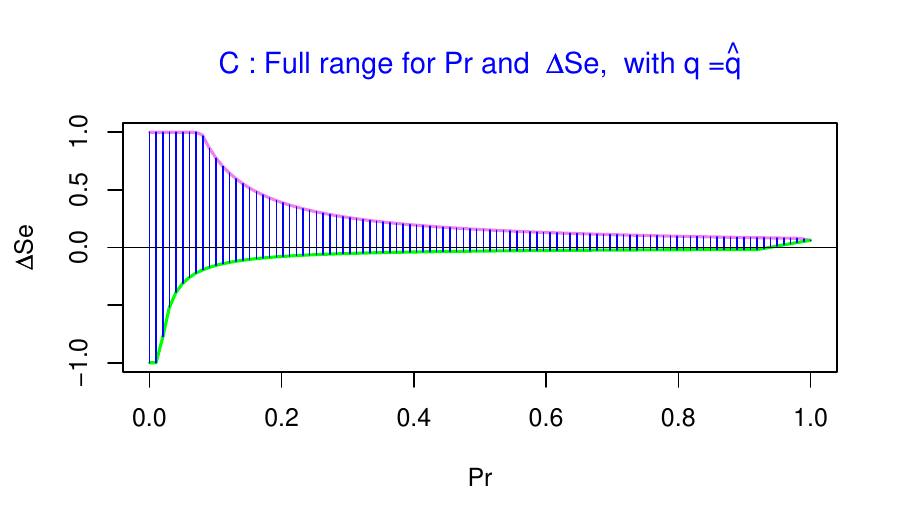}%

\begin{lem}\label{Lemma.2.4}
$  D :=\big\{\,\left(\pi^{}_1, \chi^{(1)}_{1|1}\right)\, :\,
  (\pi,\chi)\in \mu^{-1}(\{q\})  \big\} $
is the  nonempty set of all $(\Pr,\Se_1)\in[0,1]^2$ satisfying  
\eqref{eq.cond10}.
\end{lem}

\begin{lem}\label{Lemma.2.5} 
$  E :=\big\{\,\left(\pi^{}_1, \chi^{(2)}_{1|1}\right)\, :\,
  (\pi,\chi)\in \mu^{-1}(\{q\})  \big\} $
is the  nonempty set of all $(\Pr,\Se_2)\in[0,1]^2$ satisfying  
\eqref{eq.cond11}.
\end{lem}

\begin{lem}\label{Lemma.2.8}
\[
 F&:=& \big\{\,\left(\chi^{(2)}_{1|1}   - \chi^{(1)}_{1|1}\right)\, :\,
  (\pi,\chi)\in \mu^{-1}(\{q\})  \big\} \,\,\,=\,\,\, [-1,1] \\
 G&:=& \big\{\, \chi^{(1)}_{1|1} \, :\,
  (\pi,\chi)\in \mu^{-1}(\{q\})  \big\} \,\,\,=\,\,\, [0,1] \\
 H&:=& \big\{\, \chi^{(2)}_{1|1} \, :\,
  (\pi,\chi)\in \mu^{-1}(\{q\})  \big\} \,\,\,=\,\,\, [0,1] \\
 I&:=& \big\{\,\pi^{}_1 \, :\,(\pi,\chi)\in \mu^{-1}(\{q\}) 
 \big\} \,\,\,=\,\,\, [0,1]
\]
\end{lem}

\subsection{The case $d=2$ for the restricted  latent class model}\label{Subs.Aux.LCM.3}
We recall that $\mu$ denotes $\mu_2$ and that 
$q\in\prob(\{0,1\}^2)$ is fixed also in this subsection. 
Here we describe images $A_\le$ to $I_\le$
analogous to $A$ to $I$,
with  $\{\theta \in \Theta_{2,\le} : \mu(\theta)=q\}$
in place of $\{\theta \in \Theta_{2} : \mu(\theta)=q\}$.
We recall from Definition~\ref{Def:restr.lat.cl.mod}
that the subscript ``$\le$'' indicates  
that the specificity of the first test is assumed to be at most equal 
to that of the second.  Trivially, 
\[
 A_\le &:=&
\big\{\,
\left(\pi^{}_1,\chi^{(1)}_{0|0},\chi^{(1)}_{1|1},\chi^{(2)}_{0|0},\chi^{(2)}_{1|1}
  \right)
\, :\,
  (\pi,\chi)\in \mu^{-1}(\{q\}), \,  \chi^{(1)}_{0|0}\le \chi^{(2)}_{0|0}
 \big\} 
\]
is just the set of all  $(\Pr,\Sp_1,\Se_1,\Sp_2,\Se_2) \in A$
satisfying $\Sp_1\le\Sp_2$, and the nonemptyness  of this set  is proved  
at the beginning of Section~\ref{Sec.Proofs.Aux.latent.class.models} below.

\begin{lem}\label{Lemma.1.4}
$  B_\le:=\big\{\,\left(\pi^{}_1,\chi^{(1)}_{1|1},\chi^{(2)}_{1|1}\right)\, :\,
  (\pi,\chi)\in \mu^{-1}(\{q\}),\,  \chi^{(1)}_{0|0}\le \chi^{(2)}_{0|0}
   \big\} $
is the nonempty set of all $(\Pr,\Se_1,\Se_2) \in [0,1]^3$ satisfying the 
relations~\eqref{eq.cond8}, \eqref{eq.cond10}, \eqref{eq.cond11}, and 
\la
 q^{}_{01}-q^{}_{10}&\le&\Pr\,(\Se_2-\Se_1)
 \,\,\le\,\,q^{}_{01} \label{eq.cond12} 
\al
\end{lem}

\begin{lem}\label{Lemma.1.5}
$  C_\le :=
\big\{\,\left(\pi^{}_1, \chi^{(2)}_{1|1} -   \chi^{(1)}_{1|1}  \right)
 \, :\,  (\pi,\chi)\in \mu^{-1}(\{q\})
  ,\,  \chi^{(1)}_{0|0}\le \chi^{(2)}_{0|0}    \big\} $ 
is the nonempty set of all $(\Pr,\Delta\Se)\in[0,1]\times[-1,1]$ satisfying  
the  inequalities
\la                         \label{eq.cond13} 
 q^{}_{01}-q^{}_{10}  &\le& \Pr\,\Delta\Se 
  \,\,\le\,\, q^{}_{01} \wedge ( q^{}_{01}-q^{}_{10}+1-\Pr  )
\al
\end{lem}
\includegraphics[width=\linewidth,
  keepaspectratio
  ]{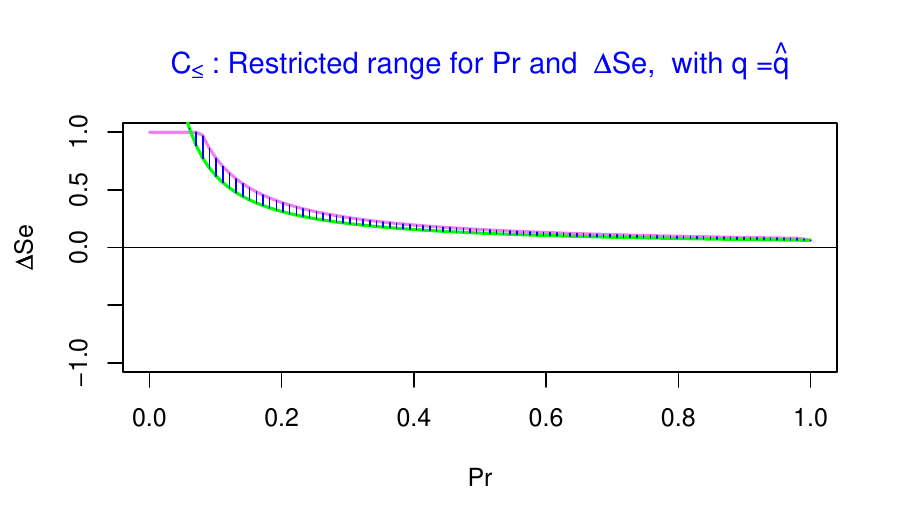}%

\begin{lem}\label{Lemma.2.10}
$ D_\le :=\big\{\,\left(\pi^{}_1,    \chi^{(1)}_{1|1}  \right)\, :\,
 (\pi,\chi)\in \mu^{-1}(\{q\})
     ,\,  \chi^{(1)}_{0|0}\le \chi^{(2)}_{0|0}   \big\} $ is 
the nonempty set of all $(\Pr,\Se_1)\in[0,1]^2$ satisfying  
the  inequalities
\la                          \label{eq.cond14A} 
 \quad\Pr-q^{}_{0+}  &\le& \Pr\, \Se_1 
 \,\,\,\le\,\,\,  q^{}_{1+} 
 \wedge\frac{\Pr + q^{}_{1+}-q^{}_{01}}{2}
 \wedge \left( \Pr + q^{}_{10}-q^{}_{01} \right)
\al
\end{lem}
\begin{lem}\label{Lemma.2.12}
$  E_\le :=\big\{\,\left(\pi^{}_1,\chi^{(2)}_{1|1}\right)\, :\,
  (\pi,\chi)\in \mu^{-1}(\{q\})
      ,\,  \chi^{(1)}_{0|0}\le \chi^{(2)}_{0|0}   \big\} $
is the nonempty set of all $(\Pr,\Se_2)\in[0,1]^2$ satisfying   the  inequalities
\la          \label{eq.cond14B} 
 \left(\Pr-q^{}_{+0}\right)  
 \vee \frac{\Pr +q^{}_{01}-q^{}_{+0} }2
 \vee \left(q^{}_{01}-q^{}_{10}\right)
 &\le& \Pr\,\Se_2
 \,\,\,\le\,\,\,   q^{}_{+1}
\al
\end{lem}
\begin{lem} \label{Lemma.2.13}
$F_\le :=   \big\{\,  \chi^{(2)}_{1|1}-   \chi^{(1)}_{1|1} \, :\,
(\pi,\chi)\in \mu^{-1}(\{q\}) ,\,  \chi^{(1)}_{0|0}\le \chi^{(2)}_{0|0}\big\}$ 
is the nonempty interval $[q^{}_{01}- q^{}_{10}, 1]$ if $q^{}_{01}-
q^{}_{10}>0$, and $[-1,1]$ if $q^{}_{01}-q^{}_{10}\le0$.
\end{lem}
\begin{lem}\label{Lemma.2.14} 
$G_\le := \big\{\,    \chi^{(1)}_{1|1}  \, :\,
 (\pi,\chi)\in \mu^{-1}(\{q\})     ,\,  \chi^{(1)}_{0|0}\le \chi^{(2)}_{0|0}
 \big\}$
is the nonempty interval
\[
 \left[\, 0\, ,\, \frac{q^{}_{1+}}{ q^{}_{1+} + q^{}_{01} } 
 \vee \left(\frac{q^{}_{11}}{\left( q^{}_{+1} - q^{}_{10}\right)^+} \wedge
   1\right)\,\right] 
\qquad \text{ with }\quad\frac 00 \,:=\, 1 
\]
\end{lem}
\begin{lem}\label{Lemma.2.15}  
$H_\le := \big\{\,    \chi^{(2)}_{1|1}  \, :\,
 (\pi,\chi)\in \mu^{-1}(\{q\})     ,\,  \chi^{(1)}_{0|0}\le \chi^{(2)}_{0|0}
 \big\}$
is the nonempty interval 
\[
 \left[\, \frac{q^{}_{01}}{q^{}_{0+}+ q^{}_{10}} \wedge
 \frac{\left(
     q^{}_{01}-q^{}_{10}\right)^+}{q^{}_{0+}-q^{}_{10}}\,\,,\,1\,\right]
\qquad \text{ with }\quad\frac 00 \,:=\, 1 
\]
\end{lem}
\begin{lem}\label{Lemma.2.16}  
$I_{\le} := \big\{\,    \pi^{}_{1}  \, :\,
 (\pi,\chi)\in \mu^{-1}(\{q\})     ,\,  \chi^{(1)}_{0|0}\le \chi^{(2)}_{0|0}
 \big\}$
is the nonempty interval
\[
 \left[\, \left(q^{}_{01} - q^{}_{10}\right)^+ \, ,\,
  \left(1-   \left(q^{}_{10} - q^{}_{01}\right)^+   \right)
  \vee \left(q^{}_{10} - q^{}_{01}\right)
 \,\right] 
\]
\end{lem}
\section{Proofs for Section~\ref{Sec.Aux.latent.class.models}}
\label{Sec.Proofs.Aux.latent.class.models}
Let us first address the nonemptyness of the sets $A$ to $I_\le$.
Below, we prove Lemmas~\ref{Basic.lemma}-\ref{Lemma.2.16} 
ignoring the word ``nonempty''. So, strictly speaking, we should 
rather write  something like ``Proof of Weak Lemma~\ref{Basic.lemma}'' 
and so on below. We next observe that, say,  the interval 
then known to equal $G_\le$  by Lemma~\ref{Lemma.2.14} 
is nonempty, as it contains zero.
Hence $A_\le$ is nonempty, since $G_\le$ is the image of 
$A_\le$ under some function. Hence $A\supseteq A_\le$  is nonempty.
Hence the remaining sets are nonempty, 
as they are  images of $A$ or  $ A_\le$ under certain functions.

\begin{proof}[Proof of Lemma~\ref{Lem.d=1}] 
If $(\pi,\chi)\in \mu_1^{-1}(\{q \})$, then we have in particular
$(\pi_1,\chi^{}_{0|0},\chi^{}_{1|1}) \in[0,1]^3$ 
and $(1-\pi_1)(1-\chi^{}_{0|0})+\pi_1\chi^{}_{1|1}=
\big(\mu_1(\pi,\chi)\big)_1=q_1$.
This shows that ``$\subseteq$'' holds in the claimed equality.
If, conversely, $(\Pr,\Sp,\Se)$ belongs to the second set,
and if we put $\pi_0:=1-\Pr$, $\pi_1:=\Pr$, $\chi^{}_{0|0}:=\Sp$,
$\chi^{}_{1|0}:= 1-\Sp$, $\chi^{}_{0|1}:=1-\Se$, $\chi^{}_{1|1}:=\Se$,
then $(\pi,\chi)\in \mu_1^{-1}(\{q \})$. Thus ``$\supseteq$'' holds as well.
\end{proof}
\begin{proof}[Proof of Lemma~\ref{Lem.d=2.d=1}]
If $(\pi,\chi)\in\mu_2^{-1}(\{q \})$, then $(\pi,\chi^{(1)})\in\Theta_1$
and for $\iota\in\{0,1\}$ we have
\[
\mu_1(\pi,\chi^{(1)} ) 
&=& \sum_{i=0}^1\pi_i\chi^{(1)}_{\iota|i} 
\,\,\,= \,\,\,  \sum_{i=0}^1\sum_{\kappa=0}^1 \pi_i \chi^{}_{\iota\kappa|i}
\,\,\,=\,\,\, \sum_{\kappa=0}^1q_{\iota\kappa}
\,\,\,=\,\,\, q_{\iota+}
\]
This proves ``$\subseteq$'' in \eqref{eq.mu2mu1.1}.
If $(\pi,\psi)\in \mu_1^{-1}(\{q_{\cdot+}\})$, then,  by 
\[
  \chi^{}_{j|i} &:=& \frac {q_j}{q_{j^{}_1+}}\psi^{}_{j^{}_1|i}
  \qquad\text{ for }(i,j)\in \{0,1\}\times\{0,1\}^2
\]
with the nonstandard convention $\frac 00:=\frac 12$,
we define a 
$\chi \in \markov(\{0,1\},\{0,1\}^2)$ with $\chi^{(1)}=\psi$ 
and
\[
 \big(\mu_2(\pi,\chi)\big)_j 
  &=& \sum_{i=0}^1\pi_i\chi^{}_{j|i} \,\,\,=\,\,\,
\frac {q_j}{q_{j^{}_1+}}\sum_{i=0}^1\pi_i \psi^{}_{j^{}_1|i}
\,\,\,=\,\,\, q_j
 \qquad\text{ for } j\in\{0,1\}^2
\]
This proves ``$\supseteq$'' in \eqref{eq.mu2mu1.1}.
The proof of \eqref{eq.mu2mu1.2} is analogous.
\end{proof} 
\begin{proof}[Proof of Lemma~\ref{Basic.lemma}]  
Call $A'$ the set claimed to equal $A$ up to line~\eqref{eq.cond4}.
If $\theta=(\pi,\chi)\in\Theta_2$, then obviously 
\la        \label{eq.pi1etc01}
(\pi^{}_1,\chi^{}_{0+|0},\chi^{}_{1+|1},\chi^{}_{+0|0},\chi^{}_{+1|1})&\in&[0,1]^5
\al
and, using \eqref{eq.def.lambda} and identities like $\pi^{}_0=1-\pi^{}_1$ 
and  $\chi^{}_{1+|0}   = 1-\chi^{}_{0+|0} $,
the condition $\mu(\theta)=q$ is seen to be equivalent to the system of
three equations 
\la
 (1-\pi^{}_1)(1-\chi^{}_{0+|0}) + \pi^{}_{1}\chi^{}_{1+|1}&=&q_{1+}\label{eq.q1+}\\
 (1-\pi^{}_1)(1-\chi^{}_{+0|0}) + \pi^{}_{1}\chi^{}_{+1|1}&=&q_{+1}\label{eq.q+1}\\
 (1-\pi^{}_1) \chi^{}_{00|0}+ \pi^{}_{1}\chi^{}_{00|1} &=& q_{00}  \label{eq.q00}
\al
Using first  $\chi^{}_{00|0} \le   \chi^{}_{0+|0} \wedge   \chi^{}_{+0|0}$
and  $\chi^{}_{00|1} \le \chi^{}_{0+|1} \wedge   \chi^{}_{+0|1}
= (1- \chi^{}_{1+|1})\wedge (1- \chi^{}_{+1|1})$
and second $  \chi^{}_{00|0}= \chi^{}_{0+|0}+\chi^{}_{+0|0}-(1-\chi^{}_{11|0})
\ge (\chi^{}_{0+|0}+\chi^{}_{+0|0}-1)^+_{} $  
and  $ \chi^{}_{00|1} =1+ \chi^{}_{11|1}- \chi^{}_{1+|1} -\chi^{}_{+1|1}
   \ge(1- \chi^{}_{1+|1}- \chi^{}_{+1|1} )^+_{}   $, we see that 
equation~\eqref{eq.q00} would imply  the two inequalities
\[
 (1-\pi^{}_1) \chi^{}_{0+|0} \wedge   \chi^{}_{+0|0} 
    +  \pi^{}_1   (1- \chi^{}_{1+|1})\wedge (1- \chi^{}_{+1|1}) &\ge & q_{00}   \\
(1-\pi^{}_1)(\chi^{}_{0+|0}+\chi^{}_{+0|0}-1)_{}^+
               + \pi^{}_1 (1- \chi^{}_{1+|1}- \chi^{}_{+1|1} )^+_{}  &\le & q_{00}
\]
Thus for $\theta\in\Theta_2$ with $\mu(\theta)=q$, the left hand side 
of~\eqref{eq.pi1etc01} belongs to $A'$.  Hence $A\subseteq A'$. 

To prove the reversed inclusion, let  $(\Pr,\Sp_1,\Se_1,\Sp_2,\Se_2) \in A'$.
Put $\pi =(\pi^{}_0, \pi^{}_1) := (1-\Pr,\Pr)$.
Choose two numbers
\[
 \chi^{}_{00|0} &\in& [ (\Sp_1+\Sp_2-1)_{}^+ ,\Sp_1 \wedge \Sp_2   ] \\
 \chi^{}_{00|1} &\in& [ (1-\Se_1-\Se_2)_{}^+ ,(1-\Se_1)\wedge(1-\Se_2)  ]
\] 
such that 
\[
  (1-\Pr )  \chi^{}_{00|0}    + \Pr\, \chi^{}_{00|1}  &=&  q^{}_{00}
\]
This is possible by connectedness, since the two intervals above are
nonempty and we would get ``$\le  q^{}_{00}$'' by choosing  the
lower endpoints and ``$\ge  q^{}_{00}$'' for the  upper ones.
Now put 
$$
\begin{array}[\displaystyle]{rclrcl}
 \chi^{}_{01|0}&:=&\Sp_1- \chi^{}_{00|0}&\qquad \chi^{}_{01|1}&:=&1-\Se_1-\chi^{}_{00|1}\\
 \chi^{}_{10|0}&:=&\Sp_2- \chi^{}_{00|0}&       \chi^{}_{10|1}&:=&1-\Se_2-\chi^{}_{00|1}\\
 \chi^{}_{11|0}&:=& 1- \Sp_1 - \Sp_2   +\chi^{}_{00|0}  
      & \chi^{}_{11|1}&:=& \Se_1 + \Se_2 -1   +\chi^{}_{00|1}  
\end{array}
$$       
Then $(\pi,\chi)\in \Theta_2$ satisfies the equations \eqref{eq.q1+}-\eqref{eq.q00},
so that   $\mu(\theta)=q$, and the 
corresponding element of $A$ is $(\Pr,\Sp_1,\Se_1,\Sp_2,\Se_2)$.
Hence we also have $A'\subseteq A$.

Obviously, equations \eqref{eq.cond1} and \eqref{eq.cond2}
are jointly equivalent to \eqref{eq.cond5} and \eqref{eq.cond6}, 
by addition and subtraction. 
In the presence of~\eqref{eq.cond1} and~\eqref{eq.cond2}, we have
\[
 (1-\Pr)\,\Sp_1\wedge\Sp_2 &=& 1-\Pr- (1-\Pr)(1-\Sp_1)\vee (1-\Sp_2)\\
  &=& 1-\Pr-(q^{}_{1+} -\Pr\,\Se_1)\vee(q^{}_{+1} -\Pr\,\Se_2)
\]
so that, by inserting and rearranging, inequality \eqref{eq.cond3}
is equivalent to 
\[ 
 \Pr\,\Se_1\vee\Se_2+(q^{}_{10}-\Pr\,\Se_1)\vee(q^{}_{01}-\Pr\,\Se_2)
 &\le&q^{}_{01}+q^{}_{10} 
\]
which,  by considering separately the four cases $a\vee b+c\vee d =a+c$ etc.,
simplifies to~\eqref{eq.cond7}.
Finally, in the presence of~\eqref{eq.cond7}, inequality  \eqref{eq.cond4}
is equivalent to $(q^{}_{00}-q^{}_{11}-x)_{}^+ + x_{}^+ \le q^{}_{00}$   
with $x:=\Pr\,(1-\Se_1-\Se_2)$, which simplifies to $-q^{}_{11}\le x \le
q^{}_{00}$, that is, \eqref{eq.cond8}.
\end{proof}
\begin{proof}[Proof of Lemma~\ref{Lemma.1.3}]  
Call $B'$ the set claimed to equal $B$. 
By Lem\-ma~\ref{Basic.lemma}, we have $(\Pr,\Se_1,\Se_2)\in B$
iff there exist $\Sp_1,\Sp_2 $ with 
the quintuple 
$(\Pr,\Sp_1,\Se_1,\Sp_2,\Se_2)\in[0,1]^5$
satisfying \eqref{eq.cond1}-\eqref{eq.cond4}
or, equivalently, \eqref{eq.cond5}-\eqref{eq.cond8}. 
So in this case, we have in particular 
  \eqref{eq.cond7} and \eqref{eq.cond8},
while \eqref{eq.cond1} and \eqref{eq.cond2} together with $\Sp_1,\Sp_2  \in[0, 1]$
yield  \eqref{eq.cond10} and \eqref{eq.cond11}. 

Conversely, if $(\Pr,\Se_1,\Se_2)\in B'$, then by \eqref{eq.cond10} and
\eqref{eq.cond11}   we can find $\Sp_1,\Sp_2\in[0,1]$ satisfying 
\eqref{eq.cond1} and \eqref{eq.cond2}, and hence
\eqref{eq.cond5} and \eqref{eq.cond6},
and thus $(\Pr,\Se_1,\Se_2)\in B$.
\end{proof}
\begin{proof}[Proof of Lemma~\ref{Lemma.3.4.new}]
Call $C'$ the set claimed to equal $C$. By Lem\-ma~\ref{Lemma.1.3}, we have
$(\Pr,\Delta\Se) \in C$ iff there exist $\Se_1,\Se_2$ with 
$(\Pr,\Se_1,\Se_2)\in [0,1]^3$ satisfying 
\eqref{eq.cond7}, \eqref{eq.cond8}, \eqref{eq.cond10}, \eqref{eq.cond11}, 
and $\Delta\Se = \Se_2- \Se_1$.

Let $(\Pr,\Delta\Se) \in C$ and let  $\Se_1,\Se_2$ be as just stated.
Then~\eqref{eq.cond10} and~\eqref{eq.cond11} yield
\la    \label{Upper.bound.PrDeltaSe}
 &&\Pr\,\Delta\Se \,=\, \Pr\, \Se_2 - \Pr\, \Se_1 
  \,\le\,  q^{}_{+1} - (\Pr- q^{}_{0+})
  \,=\,   q^{}_{01} -   q^{}_{10}     + 1-\Pr 
\al
and similarly 
\[
 &&\Pr\,\,\Delta\Se \,\,\ge \,\, \Pr-q^{}_{+0} - q^{}_{1+} 
 \, \,=\,\,q^{}_{01}-q^{}_{10} +\Pr-1 
\]
Together with \eqref{eq.cond7}, the above yields \eqref{eq.cond14}.

Conversely, let  $(\Pr,\Delta\Se) \in C'$. If $\Pr=0$, then 
we may put $\Se_1 := (\Delta\Se)^-$ and $\Se_2 := (\Delta\Se)^+$,
and observe that   $(\Pr,\Se_1,\Se_2)$ then satisfies 
\eqref{eq.cond7}, \eqref{eq.cond8}, \eqref{eq.cond10}, \eqref{eq.cond11},
and $\Delta\Se = \Se_2-\Se_1$.

So assume $\Pr>0$. By connectedness we can choose 
\[
 \Se_1 &\in& \Big[ \Big(1- \frac{q_{0+}}{\Pr}\Big)\vee 0\,,\,
 \frac{q_{1+}}{\Pr}\wedge 1  \Big] \,\,=:\,\,[a_1,b_1] \\
\Se_2 &\in& \Big[ \Big(1- \frac{q_{+0}}{\Pr}\Big)\vee 0\,,\,
 \frac{q_{+1}}{\Pr}\wedge 1  \Big] \,\,=:\,\,[a_2,b_2]
\]
in such a way that $\Se_2-\Se_1 =\Delta\Se$, since the two intervals
above are nonempty and since  taking  $\Se_1 =a_1$  and $\Se_2 =b_2$ would yield 
\[
 \Pr\, (\Se_2-\Se_1) &=& q^{}_{+1}\wedge \Pr \,-\, (\Pr - q^{}_{0+})\vee 0\\
  &=& \min\{ q^{}_{+1} -\Pr+ q^{}_{0+}\,,\, q^{}_{+1}\,,\,q^{}_{0+},\,\Pr     \} \\
  &\ge& \min\{ q^{}_{01}-q^{}_{10}+1-\Pr \,,\,  q^{}_{01}  \,,\,\Pr\,\Delta\Se \}\\
  &=& \Pr\,\Delta\Se
\]
using \eqref{eq.cond14} in the last step, 
while $\Se_1=b_1$ and $\Se_2 =a_2$  would similarly yield 
\[
 \Pr\, (\Se_2-\Se_1) &=&  (\Pr - q^{}_{+0})\vee 0 \,-\,q^{}_{1+}\wedge \Pr\\
  &=& \max \{\Pr - q^{}_{+0}- q^{}_{1+} \,,\,  - q^{}_{+0} \,,\,
      - q^{}_{1+}  \,,\,-\Pr  \} \\
  &\le & \max\{  q^{}_{01}- q^{}_{10} +\Pr-1 \,,\,-q^{}_{10}\,,\, \Pr\,\Delta\Se  \} \\
  &= & \Pr\,\Delta\Se
\]
using again \eqref{eq.cond14} in the last step.
For every choice of $\Se_1$ and $\Se_2$ 
as above, the triple $(\Pr,\Se_1,\Se_2)$ obviously satisfies
\eqref{eq.cond7}, \eqref{eq.cond10}, and \eqref{eq.cond11}.
To get~\eqref{eq.cond8} as well,
we have to refine our choice: Since the condition 
$\Se_2-\Se_1 =\Delta\Se$ is not affected by a same translation of $\Se_1$ and
$\Se_2$, we could choose $\Se_1$ and $\Se_2$ such that $\Se_i=a_i$ for some $i$,
which always yields 
$\Pr\,(\Se_1+\Se_2 -1) \le q^{}_{11}$ as in the  case of $i=1$: 
\[
 \Pr\,(\Se_1+\Se_2 -1) &\le&  
(\Pr-q^{}_{0+})\vee 0  +q^{}_{+1}\wedge \Pr  -\Pr \\
&=& \max\{ -q^{}_{0+} +  q^{}_{+1}\wedge \Pr\,,\, q^{}_{+1}\wedge \Pr  -\Pr  \}\\
&\le & \max\{ q^{}_{+1} - q^{}_{0+} \,,\,0\} \\
 &\le & q^{}_{11}
\]
Alternatively we could choose $\Se_1$ and $\Se_2$ such that  $\Se_i=b_i$ for some
$i$, yielding $\Pr\,(\Se_1+\Se_2 -1) \ge -q^{}_{00}$. By connectedness, then,
we can choose $\Se_1$ and $\Se_2$ such that \eqref{eq.cond8} holds.
Then $(\Pr,\Se_1,\Se_2)\in D$ and we get $(\Pr,\Delta\Se)=(\Pr,\Se_2-\Se_1)\in C$.
\end{proof}
\begin{proof}[Proof of Lemma~\ref{Lemma.2.4}] 
Call $D'$ the set claimed to equal $D$.
Lemma~\ref{Lem.d=2.d=1}  yields 
$ D =\{(\pi_1,\chi^{}_{1|1}) : (\pi,\chi) \in  \mu_1^{-1}(q_{\cdot+}) \}$,
which by Lemma~\ref{Lem.d=1} equals
\[
\quad
\{(\Pr,\Se_1)\in[0,1]^2 : \exists \,\Sp_1\in[0,1]\text{ with }
   (1-\mathrm{Pr})(1-\mathrm{Sp_1}) + \mathrm{Pr} \,\mathrm{Se_1} = q^{}_{1+}\}
\]
Thus, if $(\Pr,\Se_1)\in D$ and if $\Sp_1$ is
chosen according to the above, then using 
$\Sp_1\ge 0$ and $\Sp_1\le 1$ yields \eqref{eq.cond10}
and hence $ (\Pr,\Se_1)\in D'$. 
Conversely, if $ (\Pr,\Se_1)\in D'$,
then $\Sp_1 := 1-(q^{}_{1+} -\Pr\,\Se_1)/(1-\Pr) \in[0,1]$,
even if $\Pr=1$ using $0/0:=0$, hence  $ (\Pr,\Se_1)\in D$.
\end{proof}
\begin{proof}[Proof of Lemma~\ref{Lemma.2.5}] 
As above for   Lemma~\ref{Lemma.2.4}. 
\end{proof}
\begin{proof}[Proof of Lemma~\ref{Lemma.2.8}]
In each case, the ``$\subseteq$'' claim is trivially true. To prove 
``$\supseteq$'', use Lemmas~\ref{Lemma.3.4.new}, \ref{Lemma.2.4}, \ref{Lemma.2.5}
with $\Pr=0$ for $F,G,H$, and Lemma~\ref{Lemma.2.4} with $\Se_1=0$ for I.
\end{proof} 
\begin{proof}[Proof of Lemma~\ref{Lemma.1.4}]  
Call $B_\le'$ the set claimed to equal $B_\le$.
By Lem\-ma~\ref{Basic.lemma}, we have $(\Pr,\Se_1,\Se_2)\in B_\le$
iff there exist $\Sp_1,\Sp_2 $ with 
the quintuple 
$(\Pr,\Sp_1,\Se_1,\Sp_2,\Se_2)\in[0,1]^5$
satisfying 
\eqref{eq.cond1}-\eqref{eq.cond4},
or equivalently \eqref{eq.cond5}-\eqref{eq.cond8}, and additionally 
\la      \label{eq.cond.Sp1.le.Sp2}
     \Sp_1 &\le & \Sp_2
\al

Let  $(\Pr,\Se_1,\Se_2)\in B_\le$. Then Lemma~\ref{Lemma.1.3} and 
$B_\le\subseteq B$ yield \eqref{eq.cond7}-\eqref{eq.cond11}, and  
using  \eqref{eq.cond5} and \eqref{eq.cond.Sp1.le.Sp2},
we can  sharpen \eqref{eq.cond7} to  \eqref{eq.cond12}, 
so that  $(\Pr,\Se_1,\Se_2)\in B_\le'$.
 
Conversely, let  $(\Pr,\Se_1,\Se_2)\in B_\le'$. 
Then, since~\eqref{eq.cond12} implies~\eqref{eq.cond7}, 
Lem\-ma~\ref{Lemma.1.3} yields  $(\Pr,\Se_1,\Se_2)\in B$,
so that  there exist $\Sp_1,\Sp_2\in[0,1]$
such that  $(\Pr,\Sp_1,\Se_1,\Sp_2,\Se_2)\in A$.
If $\Pr=1$, then  by Lemma~\ref{Basic.lemma} we can choose
e.g.~$\Sp_1=\Sp_2 =\frac 12$, 
since \eqref{eq.cond5}-\eqref{eq.cond8} remain unaffected, 
and hence get~\eqref{eq.cond.Sp1.le.Sp2}.
If $\Pr <1$, then  \eqref{eq.cond5} and
the left hand inequality in \eqref{eq.cond12}
yield~\eqref{eq.cond.Sp1.le.Sp2}. Thus  $(\Pr,\Se_1,\Se_2)\in B_\le$.
\end{proof}
\begin{proof}[Proof of Lemma \ref{Lemma.1.5}]    
Call $C_\le'$ the set claimed to equal $C_\le$.
By Lem\-ma~\ref{Lemma.1.4}, we have  $(\Pr,\Delta\Se)\in C_\le$ iff
there exist $\Se_1,\Se_2$ with $(\Pr,\Se_1,\Se_2) \in[0,1]^3$
satisfying  \eqref{eq.cond8}, \eqref{eq.cond10}, \eqref{eq.cond11}, 
\eqref{eq.cond12}, and $\Delta\Se=\Se_2-\Se_1$. 

Let  $(\Pr,\Delta\Se)\in C_\le$ and let $\Se_1,\Se_2$ be as just stated.
Then \eqref{eq.cond10}   and \eqref{eq.cond11} yield~\eqref{Upper.bound.PrDeltaSe},
and together with \eqref{eq.cond12} this yields   \eqref{eq.cond13},
hence  $(\Pr,\Delta\Se)\in C_\le'$.

Conversely, let $(\Pr,\Delta\Se)\in C_\le'$.
Then, since~\eqref{eq.cond13} implies~\eqref{eq.cond14},
Lem\-ma~\ref{Lemma.3.4.new} yields $(\Pr,\Delta\Se)\in C$, so that there
exist $\Se_1,\Se_2$ with $(\Pr,\Se_1,\Se_2) \in B$ and 
$\Delta\Se=\Se_2-\Se_1$. Now the left hand inequality 
in~\eqref{eq.cond13} yields $q_{01}^{}-q^{}_{10} \le \Pr\,(\Se_2-\Se_1)$, 
which together with~\eqref{eq.cond7} yields~\eqref{eq.cond12}.
Hence, by Lemma~~\ref{Lemma.1.4}, we have 
$(\Pr,\Se_1,\Se_2)\in B_\le$ and thus $(\Pr,\Delta\Se)=(\Pr,\Se_2-\Se_1)\in C_\le$.
\end{proof}

\begin{proof}[Proof of Lemma~\ref{Lemma.2.10}] 
Call $D_{\le}'$ the set claimed to equal $D_{\le}$. By Lem\-ma~\ref{Lemma.1.4},
we have $(\Pr,\Se_1)\in D_{\le}$ iff there exists $\Se_2$ 
with $(\Pr,\Se_1,\Se_2)\in[0,1]^3$
satisfying~\eqref{eq.cond8}, \eqref{eq.cond10}, \eqref{eq.cond11},
and~\eqref{eq.cond12}. 

Let  $(\Pr,\Se_1)\in D_{\le}$ and let $\Se_2$ 
be as just stated.
Then~\eqref{eq.cond8} and~\eqref{eq.cond12} yield
\[
\Pr\,\Se_1 
  &=& \frac 12 \big(\Pr + \Pr\,(\Se_1+\Se_2-1) - \Pr \,(\Se_2 -\Se_1)\big) \\
  &\le& \frac 12 \big( \Pr + q^{}_{11} - (q^{}_{01} -q^{}_{10})\big)  
   \,\,\,=\,\,\, \frac{\Pr + q^{}_{1+} - q^{}_{01}  }{2}
\]
and~\eqref{eq.cond12} and $\Se_2\le1$ yield
$\Pr\,\Se_1 = \Pr\,\Se_2  -\Pr\,(\Se_2-\Se_1) 
 \le  
\Pr + q^{}_{10}-q^{}_{01}$.
Combined with~\eqref{eq.cond10}, we get~\eqref{eq.cond14A}.
Hence $(\Pr,\Se_1)\in D_{\le}'$.

Conversely, let $(\Pr,\Se_1)\in D_{\le}'$. 
Then, since~\eqref{eq.cond14A} implies~\eqref{eq.cond10},
Lemma~\ref{Lemma.2.4} yields $(\Pr, \Se_1)\in D$, 
so that there exists $\Se_2$ with  $(\Pr,\Se_1,\Se_2)\in B$.
By Lem\-ma~\ref{Lemma.1.3}, this is equivalent  
to $\Se_2$ fulfilling the conditions  $\Se_2 \in [0,1]$ and 
\eqref{eq.cond7}-\eqref{eq.cond11}, 
and we may assume that $\Se_2$ has been chosen maximal with this property.
Then at least one of the following four cases occurs, with each leading 
via~\eqref{eq.cond14A} or trivially
to $\Pr\,(\Se_2-\Se_1) \ge   q^{}_{01}-q^{}_{10}$ and hence,
using~\eqref{eq.cond7}, 
to~\eqref{eq.cond12}, proving  $(\Pr,\Se_1)\in D_{\le}$ as desired:

Case 1: $\Se_2 =1$. Then 
$\Pr\,(\Se_2-\Se_1) = \Pr -\Pr\,\Se_1  
 \ge  \Pr - ( \Pr + q^{}_{10}-q^{}_{01} )  
=   q^{}_{01}-q^{}_{10}$.
Case 2: Equality holds on the right in~\eqref{eq.cond7}.
Case 3: Equality holds on the right in~\eqref{eq.cond8}. Then 
\[
 \Pr\,(\Se_2-\Se_1)& =& \Pr +\Pr\,(\Se_1 +\Se_2-1) -2\,\Pr\,\Se_1  \\
  &\ge&  \Pr + q^{}_{11} -( \Pr + q^{}_{1+} - q^{}_{01}) 
  \,\,\,=\,\,\,  q^{}_{01}-q^{}_{10}
\]
Case 4: Equality holds on the right in~\eqref{eq.cond11}. Then 
$ \Pr\,(\Se_2-\Se_1) = \Pr\,\Se_2 - \Pr\,\Se_1
\ge   q^{}_{+1} - q^{}_{1+} =  q^{}_{01}-q^{}_{10} $.
\end{proof}
\begin{proof}[Proof of Lemma~\ref{Lemma.2.12}] 
Call $E_\le'$ the set  claimed to equal   $E_\le$. By Lem\-ma~\ref{Lemma.1.4},
we have $(\Pr,\Se_2)\in E_{\le}$ iff there exists $\Se_1$ 
with $(\Pr,\Se_1,\Se_2)\in[0,1]^3$
satisfying~\eqref{eq.cond8}, \eqref{eq.cond10}, \eqref{eq.cond11},
and~\eqref{eq.cond12}. 

Let  $(\Pr,\Se_2)\in E_{\le}$ and let $\Se_1$  be as just stated.
Then~\eqref{eq.cond8} and~\eqref{eq.cond12} yield
\[
\Pr\,\Se_2 
 &=&\frac 12 \big(\Pr + \Pr\,(\Se_1+\Se_2-1) + \Pr \,(\Se_2-\Se_1)\big) \\
 &\ge& \frac 12 \big( \Pr - q^{}_{00} + q^{}_{01} -q^{}_{10}\big)  
   \,\,\,=\,\,\, \frac{\Pr + q^{}_{01} - q^{}_{+0}}{2}
\]
and~\eqref{eq.cond12} and $\Pr\,\Se_1\ge0$ yield
$\Pr\,\Se_2 = \Pr\,(\Se_2-\Se_1) +  \Pr\,\Se_1 \ge  q^{}_{01}-q^{}_{10}$.
Combined with~\eqref{eq.cond11}, we get~\eqref{eq.cond14B}.
Hence $(\Pr,\Se_1)\in D_{\le}'$.

Conversely, let $(\Pr,\Se_2)\in E_{\le}'$. 
Then, since~\eqref{eq.cond14B} implies~\eqref{eq.cond11},
Lemma~\ref{Lemma.2.5} yields $(\Pr, \Se_2)\in E$, 
so that there exists $\Se_1$ with  $(\Pr,\Se_1,\Se_2)\in B$.
By Lem\-ma~\ref{Lemma.1.3}, this is equivalent  
to $\Se_1$ fulfilling the conditions  $\Se_1 \in [0,1]$ and 
\eqref{eq.cond7}-\eqref{eq.cond11}, 
and we may assume that $\Se_1$ has been chosen minimal with this property.
Then at least one of the following four cases occurs, with each leading 
via~\eqref{eq.cond14B} or trivially
to $\Pr\,(\Se_2-\Se_1) \ge   q^{}_{01}-q^{}_{10}$ and hence,
using~\eqref{eq.cond7}, 
to~\eqref{eq.cond12}, proving  
$(\Pr,\Se_2)\in E_{\le}$ as desired:

Case 1: $\Se_1 =0$. Then 
$\Pr\,(\Se_2-\Se_1) = \Pr\,\Se_2  \ge  q^{}_{01}-q^{}_{10}$.
Case 2: Equality holds on the right in~\eqref{eq.cond7}.
Case 3: Equality holds on the left in~\eqref{eq.cond8}. Then 
\[
 \Pr\,(\Se_2-\Se_1)& =& -\Pr -\Pr\,(\Se_1 +\Se_2-1) +2\,\Pr\,\Se_2  \\
  &\ge&  -\Pr + q^{}_{00} +( \Pr + q^{}_{01} - q^{}_{+0}) 
  \,\,\,=\,\,\,  q^{}_{01}-q^{}_{10}
\]
Case 4: Equality holds on the left in~\eqref{eq.cond10}. Then 
$ \Pr\,(\Se_2-\Se_1) = \Pr\,\Se_2 - \Pr\,\Se_1
 \ge  (\Pr-  q^{}_{+0})  - ( \Pr - q^{}_{0+})
=  q^{}_{01}-q^{}_{10} $.
\end{proof}

\begin{proof}[Proof of Lemma~\ref{Lemma.2.13}]
Call $F_\le'$ the interval  claimed to equal   $F_\le$.  
By Lemma~\ref{Lemma.1.5}, we  have $\Delta\Se \in F_{\le}$ iff there exists $\Pr $
with $(\Pr,\Delta\Se) \in [0,1]\times[-1,1]$ satisfying~\eqref{eq.cond13}. 

Let $\Delta\Se  \in F_{\le}$ and let $\Pr$ be as just stated. 
If $q^{}_{01}-q^{}_{10} >0$, then~\eqref{eq.cond13} yields $\Pr>0$ and hence
$\Delta\Se \ge (q^{}_{01}-q^{}_{10})/\Pr \ge q^{}_{01}-q^{}_{10}$. 
Hence  always $\Delta\Se\in F_\le'$.

Conversely, let  $\Delta\Se \in F_{\le}'$. If $q^{}_{01}-q^{}_{10} >0$, 
then $\Pr :=  (q^{}_{01}-q^{}_{10})/\Delta\Se \in{]0,1]}$
satisfies~\eqref{eq.cond13}.
If  $q^{}_{01}-q^{}_{10} \le 0$, then $\Pr:=0$ satisfies~\eqref{eq.cond13}.
Hence $\Delta\Se\in F_\le$.
\end{proof} 

\begin{proof}[Proof of Lemma~\ref{Lemma.2.14}] 
Call $G_\le'$ the interval  claimed to equal   $G_\le$. 
By Lemma~\ref{Lemma.2.10}, we  have $\Se_1 \in G_\le$ iff there exists $\Pr$
with $(\Pr,\Se_1) \in [0,1]^2$ satisfying~\eqref{eq.cond14A}. 

If $q^{}_{01}- q^{}_{10} \le 0$, then   $q^{}_{+1}-q^{}_{10}  \le  q^{}_{11}$
and hence $ G_\le' =[0,1]$; and given $\Se_1\in [0,1]$, we may put  $\Pr:=0$ 
to satisfy~\eqref{eq.cond14A}, so that also   $G_\le =[0,1]$.

So let $q^{}_{01}- q^{}_{10} >0$ for the rest of this proof. The three
functions $f_i:{]0,1]}\rightarrow \R$ defined by
\[
\qquad f_1(x) \,:=\, \frac{q^{}_{1+}}{x} \qquad 
 f_2(x) \,:=\, \frac 12 +\frac{q^{}_{1+}-q^{}_{01}}{2\,x}       \qquad 
 f_3(x) \,:=\, 1- \frac{q^{}_{01}-q^{}_{10}}{x}
\]
are continuous and monotone with  $\lim_{x\rightarrow 0}f_3(x)=-\infty$, 
so that their pointwise infimum $f:=f_1\wedge f_2 \wedge f_3$ attains 
its maximal value  at $\Pr_1:=1$ or at
some  $x \in{]0,1]}$ satisfying $f_i(x) = f_j(x)$ 
with $i< j$. The latter three equations have the unique solutions
$\Pr_{12} := q^{}_{1+} + q^{}_{01}$, $\Pr_{13} := q^{}_{+1}$,  
$\Pr_{23} := q^{}_{+1}  -  q^{}_{10}$, 
each  strictly positive by $q^{}_{01}- q^{}_{10} >0$,
and we get 
\[
 f(\Pr_1) &=& q^{}_{1+}\wedge\frac{1+q^{}_{1+} -q^{}_{01} }{2}
     \wedge \big(1+q^{}_{10} -q^{}_{01}\big)  \,\,\,= \,\,\,q^{}_{1+}  \\
 f(\Pr_{12}) &=& \frac{q^{}_{1+}}{q^{}_{1+}+q^{}_{01}}
     \wedge \left(1-\frac{q^{}_{01}-q^{}_{10}}{q^{}_{1+} + q^{}_{01}}\right)  
    \,\,\,= \,\,\, \frac{q^{}_{1+}}{q^{}_{1+}+q^{}_{01}}      \\
 f(\Pr_{13}) &=& \frac{q^{}_{1+}}{q^{}_{+1}} 
 \wedge  \left(\frac12 + \frac{q^{}_{1+}-q^{}_{01}}{2\,q^{}_{+1} } \right) 
   \,\,\,= \,\,\,\frac{q^{}_{1+} +q^{}_{11} }{ 2\,q^{}_{+1}}    \\
 f(\Pr_{23}) &=&  \frac{q^{}_{1+}}{ q^{}_{+1} - q^{}_{10}} 
  \wedge\left(1- \frac{q^{}_{01}-q^{}_{10}}{ q^{}_{+1}-q^{}_{10}}\right) 
 \,\,\,= \,\,\, \frac{q^{}_{11}}{ q^{}_{+1} - q^{}_{10}} 
\]
We have $ f(\Pr_1) \le  f(\Pr_{12})$ since $ q^{}_{1+} + q^{}_{01}\le 1$.
Writing here  $a\sim b$ to indicate that $ab>0$ or $a=b=0$ holds, 
clearing fractions yields
\[
 f(\Pr_{12})- f(\Pr_{13})
 &\sim& 2\,q^{}_{+1} q^{}_{1+}-(q^{}_{1+}+q^{}_{01}) ( q^{}_{1+} +q^{}_{11})  
 \,\,\,= \,\,\, q^{}_{01}\,( q^{}_{01}- q^{}_{1+}  )  \\
   f(\Pr_{23})- f(\Pr_{13})
&\sim& 2\,q^{}_{+1}q^{}_{11} -(q^{}_{+1} - q^{}_{10}) (q^{}_{1+} +q^{}_{11}) 
 \,\,\,= \,\,\, q^{}_{01}\,( q^{}_{1+} - q^{}_{01} ) 
\]
Hence 
$ \max f =f(\Pr_{12}) \vee  f(\Pr_{23}) $, 
namely attained  at  $\Pr_{12}$ if  $q^{}_{01}\ge q^{}_{1+}$
and at $ \Pr_{23}$ if  $q^{}_{01}\le q^{}_{1+}$.

Now let  $\Se_1\in G_\le$ and let  $\Pr\in [0,1]$ with~\eqref{eq.cond14A}.  
Then in particular $\Pr\,\Se_1 \le \Pr + q^{}_{10} -q^{}_{01}$ 
and thus  $\Pr>0 $, using  $q^{}_{01}- q^{}_{10} >0$.
Thus~\eqref{eq.cond14A} yields 
$\Se_1 \le f(\Pr) \le \max f$ and thus   $\Se_1\in G_\le'$.

Conversely, let $\Se_1\in G_\le'$.  Then  $\Se_1 \le \max f$ and thus,
by  $\lim_{x\rightarrow 0}f(x)=-\infty$ and continuity of $f$,
there is a $\Pr \in{]0,1]}$ with $f(\Pr)=\Se_1$, yielding 
$\Pr\,\Se_1= \Pr\,f(\Pr)=\mathrm{R.H.S.\eqref{eq.cond14A}}
\ge\mathrm{L.H.S.\eqref{eq.cond14A}}$, with the last inequality
due to $\Pr \le 1$ and $q\in\prob(\{0,1\}^2)$, 
so that~\eqref{eq.cond14A} holds and hence $\Se_1\in G_\le$.
\end{proof}

\begin{proof}[Proof of Lemma~\ref{Lemma.2.15}]
Very similar to the above proof of Lem\-ma~\ref{Lemma.2.14}, 
with the following differences:
Use  Lemma~\ref{Lemma.2.12} in place of Lemma~\ref{Lemma.2.10}.
After again restricting attention to the main case where
$q^{}_{01}- q^{}_{10} >0$, define now
\[
\qquad f_1(x) \,:=\, 1- \frac{q^{}_{+0}}{x} \qquad 
 f_2(x) \,:=\, \frac 12 +\frac{q^{}_{01}-q^{}_{+0}}{2\,x}  \qquad 
 f_3(x) \,:=\, \frac{q^{}_{01}-q^{}_{10}}{x}
\]
and observe that $f:=f_1\vee f_2 \vee f_3$ is minimized over $]0,1]$ 
at one of  $\Pr_1:=1$, $\Pr_{12} := q^{}_{+0} + q^{}_{01}$, 
$\Pr_{13} := q^{}_{0+}$,  
$\Pr_{23} := q^{}_{0+}  -  q^{}_{10}$.
After computing  $f(\Pr_1)= q^{}_{+1}$,  
$ f(\Pr_{12})=\frac{q^{}_{01}}{q^{}_{+0} + q^{}_{01}   }$,
$ f(\Pr_{13})=\frac{2\, q^{}_{01} -q^{}_{10} }{2\, q^{}_{0+} }$,
$ f(\Pr_{23}) =  \frac{q^{}_{01}-q^{}_{10}}{q^{}_{0+}  -  q^{}_{10} }$,
one observes $f(\Pr_1)\ge f(\Pr_{12})$ and 
\[
 f(\Pr_{12})- f(\Pr_{13})  &\sim&  q^{}_{10}\,( q^{}_{+0} - q^{}_{01} )
\,\,\,\sim\,\,\,  f(\Pr_{13})- f(\Pr_{23})
\]
Hence $  \min f = f(\Pr_{12}) \wedge  f(\Pr_{23})$.
\end{proof} 

\begin{proof}[Proof of Lemma~\ref{Lemma.2.16}] 
Call $I_{\le}'$ the interval  claimed to equal   $I_{\le}$. 
By Lemma~\ref{Lemma.1.5}, we  have $\Pr \in I_{\le}$ iff there exists $\Delta\Se $
with $(\Pr,\Delta\Se) \in [0,1]\times[-1,1]$ satisfying~\eqref{eq.cond13}. 

Let  $\Pr \in I_{\le}$ and let  $\Delta\Se$ be as just stated.
If  $q^{}_{01}- q^{}_{10}=0$, then $I_\le'=[0,1]$ and hence $\Pr\in I_\le'$.
If  $q^{}_{01}- q^{}_{10} >0$, then~\eqref{eq.cond13} implies $\Delta\Se>0$
and hence $\Pr \ge  (q^{}_{01}- q^{}_{10})/\Delta\Se \ge q^{}_{01}- q^{}_{10}$
and thus again $\Pr\in I_\le'$.
If  $q^{}_{01}- q^{}_{10} <0$, then either $q^{}_{01}- q^{}_{10}+1-\Pr \ge0$ 
and then $\Pr\le 1-( q^{}_{10}-q^{}_{01})^+$,
or $q^{}_{01}- q^{}_{10}+1-\Pr <0$  and then~\eqref{eq.cond13} yields
$\Delta\Se <0$ and thus  $\Pr\le(q^{}_{01}- q^{}_{10})/\Delta\Se \le
q^{}_{10}-q^{}_{01}$, so that  $\Pr\in I_\le'$ also in this case.

Conversely, let  $\Pr \in I_{\le}'$. If $q^{}_{01}- q^{}_{10} \le 0$, then either
$q^{}_{01}- q^{}_{10}+1-\Pr \ge0$  and then $\Delta\Se:= 0$
satisfies~\eqref{eq.cond13},
or  $q^{}_{01}- q^{}_{10}+1-\Pr <0$ and then $\Pr \in I_{\le}'$ yields 
$\Pr \le q^{}_{10}-q^{}_{01}$ so that  $\Delta\Se:= -1$
satisfies~\eqref{eq.cond13}.
If  $q^{}_{01}- q^{}_{10} > 0$, then  $\Delta\Se:=(q^{}_{01}- q^{}_{10})/\Pr
\in{]0,1]}$ satisfies~\eqref{eq.cond13}. Hence $\Pr \in I_{\le}$ in every case. 
\end{proof}

\section{The remaining proofs}\label{Sec.remaining.proofs}
Below we prove  Theorems~\ref{Thm.main.result} and 
Theorem~\ref{Thm.bound.Se1.new} by applying the rather general and trivial
Lem\-mas~\ref{Lem.comp.cb} and~\ref{Lem:NEF.Conf.bounds.equiv.equal}
together with the special Lemmas~\ref{Basic.lemma}, \ref{Lemma.1.5}, 
and~\ref{Lemma.2.12}. We then deduce Theorem~\ref{Thm.bound.Se1.new} from 
Theorem~\ref{Thm.bound.Se1}.

\begin{lem}          \label{Lem.comp.cb}
Let $\cP=(P_\theta : \theta\in\Theta)$ and 
$\cQ=(Q_\eta: \eta\in\Eta)$ be 
experiments on the same 
sample space $\cX$, with parameters of interest
$\kappa:\Theta\rightarrow\overline{\R}$
and $\lambda:\Eta\rightarrow\overline{\R}$.  Let $\beta\in[0,1]$.

{\sc A.} Assume the implication
\la          \label{Lem.comp.cb.eq1}
  \eta\in\Eta&\Rightarrow& \exists\theta\in\Theta
  \text{ \rm with }P_\theta= Q_\eta \text{ \rm and }\kappa(\theta) \le \lambda(\eta) 
\al
Then every lower $\beta$-confidence bound for $(\cP,\kappa)$ is also one
for $(\cQ,\lambda)$. 

{\sc B.}  Assume the implication
\la           \label{Lem.comp.cb.eq2}
  \eta\in\Eta&\Rightarrow& \exists\theta\in\Theta
  \text{ \rm with }P_\theta= Q_\eta \text{ \rm and }
\kappa(\theta) \ge \lambda(\eta) 
\al
and let $\underline{\kappa}$ and $\underset{\widetilde{}}{\kappa}$ be 
both lower $\beta$-confidence bounds for $(\cP,\kappa)$ and for
$(\cQ,\lambda)$, with $\underset{\widetilde{}}{\kappa}$ worse than 
$\underline{\kappa}$ for $(\cP,\kappa)$. 
Then $\underset{\widetilde{}}{\kappa}$ is also worse than 
$\underline{\kappa}$ for $(\cQ,\lambda)$.

Analogously for upper confidence bounds, with  ``$\le$'' 
and ``$\ge$'' interchanged.
\end{lem}
\begin{proof} A. Let $\underline{\kappa}$ be a lower 
$\beta$-confidence bound for $(\cP,\kappa)$ and let $\eta \in\Eta$.
With $\theta$ from~\eqref{Lem.comp.cb.eq1} then 
$Q_\eta(\underline{\kappa} \le\lambda(\eta) )
 \,=\, P_\theta(\underline{\kappa} \le\lambda(\eta) )
\,\ge\,  P_\theta(\underline{\kappa} \le\kappa(\theta) ) \,\ge\, \beta $.

B. Let $\eta \in\Eta$ and $t<\lambda(\eta)$. 
With $\theta$ from~\eqref{Lem.comp.cb.eq2} we then have $t<\kappa(\theta)$ and hence
$Q_\eta(    \underset{\widetilde{}}{\kappa}       \ge t )
\,=\, P_\theta( \underset{\widetilde{}}{\kappa}         \ge t )
\,\le\,  P_\theta( \underline{\kappa} \ge t )
\,=\, Q_\eta ( \underline{\kappa} \ge t )$.
\end{proof}

Below, a \defn{natural exponential family}, or \defn{NEF} for short,
is any statistical model $\cQ=(Q_\eta: \eta\in\Eta)$ such that, for some 
$k\in \N$ and some measure $\nu$ on $\R^k$, we have $\Eta\subseteq\R^k$ 
and, for each $\eta\in\Eta$, $Q_\eta$ is a law on $\R^k$ with a $\nu$-density
proportional to  
$y\mapsto \exp(\sum_{i=1}^k \eta_iy_i)$.

\begin{lem}  \label{Lem:NEF.Conf.bounds.equiv.equal}
Let $\cQ =(Q_\eta: \eta\in\Eta)$ be a NEF with $\Eta$ open and nonempty.
Let $\lambda:\Eta\rightarrow\overline{\R}$ be lower semicontinuous and let 
$\underline{\lambda}$ and $\underset{\widetilde{}}{\lambda}$ be  equivalent 
lower confidence bounds for $(\cQ,\lambda)$.
Then $ \underline{\lambda}\wedge \sup \lambda(\Eta)
=\underset{\widetilde{}}{\lambda}\wedge\sup\lambda(\Eta)$ $\cQ$-a.s.
\end{lem}
\begin{proof} The equivalence assumption yields 
\la       \label{lambda.bar.tilde.equiv}
 Q_\eta \big( \underline{\lambda} > t\big) 
 &=&  Q_\eta \big(\underset{\widetilde{}}{\lambda}> t\big) 
 \qquad \text{ if $\eta\in\Eta$ and $t\in{[-\infty, \lambda(\eta)[}$}
\al
For fixed $t\in\R$ with $t <\sup\lambda(\Eta)$, the subfamily
$(Q_\eta : \eta \in \Eta, \lambda(\eta) >t)$ is again a NEF with nonempty open
parameter space, hence complete in the sense of Lehmann-Scheff\'e, so that 
\eqref{lambda.bar.tilde.equiv} yields 
$ \{ \underline{\lambda} > t\} =  \{\underset{\widetilde{}}{\lambda}> t\}$ 
$\cQ$-a.s. Hence 
\[
 \{  \underline{\lambda}\wedge\sup\lambda(\Eta) \neq 
\underset{\widetilde{}}{\lambda}\wedge\sup\lambda(\Eta)\}
 &=& \bigcup_{t\in\Q,\, t <\sup\lambda(\Eta)} 
 \{\underset{\widetilde{}}{\lambda} \le t<\underline{\lambda}\}
 \cup\{\underline{\lambda}\le t<\underset{\widetilde{}}{\lambda}\}
\]
is a $\cQ$-null set.
\end{proof}

\begin{proof}[Proof of Theorem~\ref{Thm.main.result}]  
We first check the applicability of Lem\-ma~\ref{Lem.comp.cb} to
some  pairs of estimation problems.
Recall $\mu$ and $P_\theta$ from~\eqref{eq.def.mu}
and~\eqref{Eq:Def.P.theta}.

The problems $(\cM, \eqref{Param.interest.multinomial})$ and 
$(\cP_{2}, \eqref{Param.interest.full})$, in this order but also in the
reversed one, fulfill the assumptions of Lemma~\ref{Lem.comp.cb} A and B:
For the stated order, given  $\theta=(\pi,\chi)\in\Theta_2$, 
put $q:=\mu(\theta)\in\prob(\{0,1\}^2)$, and observe that then
$ \mathrm{M}_{n,q}= P_\theta$ and 
$\text{R.H.S.\eqref{Param.interest.multinomial}} = 
\text{R.H.S.\eqref{Param.interest.full}}$
by Lemma~\ref{Basic.lemma}\eqref{eq.cond5}. 
For the reversed order, given  $q\in\prob(\{0,1\}^2)$, choose
$\theta \in\Theta_2$ with $\mu(\theta)=q$ using the nonemptyness claim 
of Lemma~\ref{Basic.lemma}, and finish as in the preceding sentence.

The problems $(\cP_{2}, \eqref{Param.interest.full})$ 
and $(\cP_{2,\le},\eqref{Param.interest.restr})$
fulfill the assumptions of Lem\-ma~\ref{Lem.comp.cb}~A, since
for $\theta \in\Theta_{2,\le}$, 
we also have $\theta \in \Theta_{2}$
and $\text{R.H.S.\eqref{Param.interest.full}} \le 
\text{R.H.S.\eqref{Param.interest.restr}}$.

The problems $(\cP_{2,\le}, \eqref{Param.interest.restr})$ and 
$(\cM, \eqref{Param.interest.multinomial})$
fulfill the assumptions of Lem\-ma~\ref{Lem.comp.cb}~A and B:
Given  $q\in\prob(\{0,1\}^2)$, Lemma~\ref{Lemma.1.5}
with $\Pr=1$ yields
a $\theta =(\pi,\chi) \in \Theta_{2,\le}$ with $\mu(\theta)=q$,
so $P_\theta =  \mathrm{M}_{n,q}$, and 
$\text{R.H.S.\eqref{Param.interest.restr}}
=\text{R.H.S.\eqref{Param.interest.multinomial}}$.

Applying now Lem\-ma~\ref{Lem.comp.cb} several times
yields parts A and B of the theorem. 
The subclaim of Part~C referring only to 
$(\cM,\eqref{Param.interest.multinomial})$  and
$(\cP_{2}, \eqref{Param.interest.full})$ 
follows directly from Parts A and B, 
as ``$\underline{\Delta}$ strictly worse than
$\underset{\widetilde{}}{\Delta}$''
is equivalent to ``$\underline{\Delta}$ worse
than $\underset{\widetilde{}}{\Delta}$,
and not  $\underset{\widetilde{}}{\Delta}$ worse than $\underline{\Delta}$''.

Finally, let $\underline{\Delta}$ be admissible as a lower $\beta$-confidence 
bound for $(\cM,\eqref{Param.interest.multinomial})$. By Part~A,
$\underline{\Delta}$ is also a $\beta$-confidence bound for
$(\cP_{2,\le},\eqref{Param.interest.restr})$. Let 
$\underset{\widetilde{}}{\Delta}$ be a better $\beta$-confidence bound for
$(\cP_{2,\le},\eqref{Param.interest.restr})$.
We have to show that  $\underset{\widetilde{}}{\Delta}$
is equivalent to $\underline{\Delta}$ 
for $(\cP_{2,\le},\eqref{Param.interest.restr})$.
By Part~B,  $\underset{\widetilde{}}{\Delta}$ is better 
than $\underline{\Delta}$  also for $(\cM,\eqref{Param.interest.multinomial})$
and hence, by the assumed admissibility, in fact equivalent 
to   $\underline{\Delta}$ for $(\cM,\eqref{Param.interest.multinomial})$.
With a view towards applying Lemma~\ref{Lem:NEF.Conf.bounds.equiv.equal},
we put $\Eta:=\{\eta\in{]-\infty,0[^3} : \sum_{i=1}^3 \mathrm{e}^{\eta_i^{}}
<1\}$, define a function $\tau : \Eta \rightarrow \prob(\{0,1\}^2)$ by 
$\tau_{00}(\eta):=  \mathrm{e}^{\eta_1^{}}$,
$\tau_{01}(\eta):=  \mathrm{e}^{\eta_2^{}}$, 
$\tau_{10}(\eta):=  \mathrm{e}^{\eta_3^{}}$,
and  $\tau_{11}(\eta):= 1- \sum_{i=1}^3 \mathrm{e}^{\eta_i^{}}$ 
for $\eta\in\Eta$, and put 
$\cQ := \cM\circ \tau$, that is, $\cQ=( Q_\eta : \eta \in\Eta)$ 
with $Q_\eta=\mathrm{M}_{n,\tau(\eta)}$.
Let $\kappa$ denote the function~\eqref{Param.interest.multinomial}
and let $\lambda := \kappa\circ \tau$ so that, 
writing $\Prob(\cX)$ for the set of all laws 
on $\cX:=\{k\in\N_0^{\{0,1\}^2}: k_{++}=n\}$, the diagram
\begin{diagram}
             &             &\Eta            &         &  \\
             &\ldTo^\lambda&\dTo_\tau       &\rdTo^\cQ&  \\
\overline{\R}&\lTo_\kappa  &\prob(\{0,1\}^2)&\rTo_\cM 
&\Prob(\cX) 
\end{diagram}
commutes.
Then, trivially, 
$\underline{\lambda} :=\underline{\Delta}$ and 
$ \underset{\widetilde{}}{\lambda}:= \underset{\widetilde{}}{\Delta}$
are equivalent lower confidence bounds for $(\cQ,\lambda)$. 
Now Lemma~\ref{Lem:NEF.Conf.bounds.equiv.equal} applies
and yields $\underline{\Delta}\wedge1=
\underset{\widetilde{}}{\Delta}\wedge1$ everywhere on $\cX$ and hence,
as \eqref{Param.interest.restr} is $[-1,1]$-valued,  
the wanted equivalence.
\end{proof}
\begin{proof}[Proof of Theorem~\ref{Thm.bound.Se1.new}]
The problems $(\cM,\eqref{Param.interest.multinomial.2})$
and $(\cP_{2,\le}, \eqref{Param.interest.restr.2})$ fulfill the assumptions
of the ``upper'' version of Lemma~\ref{Lem.comp.cb}~A, since 
for  $\theta \in\Theta_{2,\le}$ and $q:= \mu(\theta)$, we have
$\text{R.H.S.\eqref{Param.interest.multinomial.2}} \ge
\text{R.H.S.\eqref{Param.interest.restr.2}}$
by Lemma~\ref{Lemma.2.14}.  
The problems $(\cP_{2,\le}, \eqref{Param.interest.restr.2})$ and 
$(\cM,\eqref{Param.interest.multinomial.2})$ fulfill the assumptions
of the ``upper'' version of Lemma~\ref{Lem.comp.cb}~A and~B, since 
for $q\in\prob(\{0,1\}^2)$, Lemma~\ref{Lemma.2.14} yields a 
$\theta\in\Theta_{2,\le}$ with $\mu(\theta)=q$ and 
$\text{R.H.S.\eqref{Param.interest.restr.2}}
=\text{R.H.S.\eqref{Param.interest.multinomial.2}}$. 
Hence Lemma~\ref{Lem.comp.cb} yields parts A and B of the theorem.

To prove Part~C, we can proceed as in the last paragraph of our 
proof of Theorem~\ref{Thm.main.result}, with the following changes:
Given now $\overline{S}$ and $\widetilde{S}$, 
we let $\kappa$ denote the function~\eqref{Param.interest.multinomial.2}.
Then Lemma~\ref{Lem:NEF.Conf.bounds.equiv.equal}
applies with $\underline{\lambda}:=-\overline{S}$ 
and $\underset{\widetilde{}}{\lambda}:=-\widetilde{S}$
to yield $\overline{S}\vee 0 = \widetilde{S}\vee0$.
Here $-\kappa\circ\eta$ is indeed lower semicontinuous, but one could also 
 replace $\Eta$ by $\{\eta\in\Eta:
\eta_2^{}>\eta_3^{}\}$; then  $-\kappa\circ\eta$ would be continuous.
\end{proof}
\begin{proof}[Proof of Theorem~\ref{Thm.bound.Se1}]
The function~\eqref{u.without.k00} is an upper $\beta$-confidence bound 
in the quadrinomial model 
$\cM := (\mathrm{M}_{n,q} : q \in \prob(\{0,1\}^2))$
and for the parameter~\eqref{Param.interest.multinomial.2}
from Theorem~\ref{Thm.bound.Se1.new},
since for $q \in \prob(\{0,1\}^2)$, conditioning on the upper left
corner of our $2\times 2$ table yields 
\[
&& \mathrm{M}_{n,q}\left(\left\{ 
 k \in\N_0^{\{0,1\}^2} : k_{++}=n,\, u(k^{}_{10},k^{}_{01},k^{}_{11}) 
\ge \text{R.H.S.\eqref{Param.interest.multinomial.2}}
  \right\}\right) \\
&=&\sum_{m=0}^n\mathrm{b}_{n,q^{}_{00}}(n-m)
 \,\mathrm{M}_{m,p}\left(\left\{j\in\N_0^{\{1,2,3\}} : j^{}_+=m, \,
 u(j) \ge
\text{R.H.S.\eqref{Eq:Second.trinomial.parameter}} 
\right\}\right) \\
&\ge&\beta
\]
with $\mathrm{b}_{n,q_{00}^{}}$ denoting  a binomial density, 
and with $p\in\prob(\{1,2,3\})$ defined by 
$p := (1-q^{}_{00})^{-1}(q^{}_{10}, q^{}_{01},q^{}_{11})$
if $q^{}_{00} <1$, and $p:=(0,0,1)$ if $q^{}_{00}=1$.
Hence the  claim follows from Theorem~\ref{Thm.bound.Se1.new}~A.
\end{proof}
\section*{Acknowledgements}
We thank our anonymous associate  editor for helpful suggestions,
Paul Taylor for making his commutative diagrams package 
available at {\tt www.PaulTaylor.EU/diagrams}, 
and Todor Dinev for carefully reading  several versions of this paper.


\begin{thebibliography}{99}
\bibitem{Abel.1993}
\textsc{Abel, U.} (1993). 
\textit{Die Bewertung diagnostischer Tests.}
Hippokrates Verlag, Stuttgart.
\bibitem{Gart.Buck.1966}
\textsc{Gart, J.J.} and \textsc{Buck, A.A.} (1966).
Comparison of a screening test and a reference test in epidemiologic
studies II: A probabilistic model for the comparison of diagnostic tests.
\textit{American Journal of Epidemiology} \textbf{83}, 593-602.
\bibitem{Lehmann.Romano.2005}
\textsc{Lehmann, E.L.} and \textsc{Romano, J.P.} (2005).
\textit{Testing Statistical Hypotheses.} Third Edition. Springer, N.Y. 
\bibitem{Leisenring.Alonzo.Pepe.2000}
\textsc{Leisenring, W.}, \textsc{Alonzo, T.} and \textsc{Pepe, M.S.} (2000).
Comparisons of predictive values of binary medical diagnostic tests 
for paired designs. \textit{Biometrics} \textbf{56}, 345-351.

\bibitem{Lloyd.Moldovan.2007}
\textsc{Lloyd, C.J.} and \textsc{Moldovan, M.V.} (2007). Exact one-sided
confidence limits for the difference between two correlated proportions.
\textit{Statistics in Medicine} \textbf{26}, 3369-3384.
A corresponding {\tt R} program named \verb}sm_file_SIM2708_2} 
is available for free at 
{\tt http://onlinelibrary.wiley.com/doi/10.1002/sim.2708/suppinfo}

\bibitem{Mattner.Winterfeld.Mattner.2009}
\textsc{Mattner, F.}, \textsc{Winterfeld, I.} and \textsc{Mattner, L.}
(2009). Sensitivit\"atsgewinn beim Testen auf toxigene C.~difficile
durch drei kommerzielle Kulturmedien.
\textit{Der Mikrobiologe} \textbf{19}, 171-176. 

\bibitem{Mattner.Winterfeld.Mattner.2010}
\textsc{Mattner, F.}, \textsc{Winterfeld, I.} and \textsc{Mattner, L.} (2012).
Diagnosing toxigenic \textit{C.~difficile}:
New confidence bounds show culturing increases sensitivity of 
toxin A/B EIA, and refute gold standards.
\textit{Scandinavian Journal of Infectious Diseases} (accepted January 2012).

\bibitem{Pepe.2003} 
\textsc{Pepe, M.S.} (2003). \textit{The Statistical Evaluation of Medical
  Tests for Classification and Prediction.} Oxford University Press, Oxford.

\bibitem{Pfanzagl.1994}
\textsc{Pfanzagl, J.} (1994).
\textit{Parametric Statistical Theory.}
de Gruyter, Berlin.

\bibitem{Weinert.et.al.1979}
\textsc{Weinert, D.A.}, \textsc{Ryan, T.J.}, \textsc{McCabe, C.H.},
\textsc{Kennedy, J.W.}, \textsc{Schloss, M.}, \textsc{Tristani, F.}, 
\textsc{Chaitman, B.R.} and \textsc{Fisher, L.D.}.
(1979). Correlations among history of angina, ST-segment response and 
prevalence of coronary-artery disease in the Coronary Artery 
Surgery Study (CASS).
\textit{New England Journal of Medicine} \textbf{301}, 230-235.

\bibitem{Zhou.McClish.Obu.2002}
\textsc{Zhou, X.H.}, \textsc{McClish, D.K.} and 
\textsc{Obuchowski, N.A.} (2002).
\textit{Statistical Methods in Diagnostic Medicine.}
Wiley, N.Y.
\end{thebibliography}
\end{document}